\documentclass[11pt]{article}
\usepackage{amsmath,amssymb,amsthm}
\usepackage{thm-restate}
\usepackage{hyperref}
\usepackage[capitalise]{cleveref}
\crefalias{subsection}{section}
\crefalias{subsubsection}{section}

\usepackage[framemethod=default]{mdframed}
\usepackage{xspace}

\usepackage{wrapfig}
\usepackage{dsfont}
\usepackage{mathtools}
\usepackage{nicefrac}
\usepackage{xcolor}
\usepackage{geometry}
\usepackage{tikz}

\usetikzlibrary{backgrounds}

\usepackage{listings}
\usepackage{float}
\floatname{algoflt}{Protocol}
\floatstyle{ruled}
\newfloat{algoflt}{tbp}{loa}
\makeatletter
\lstnewenvironment{algorithm}[2][H]{\lstset{mathescape=true,tabsize=2,numbers=none,basicstyle=\ttfamily\footnotesize,columns=fullflexible,emph={if,else},
    emphstyle=\bfseries,
    }\setlength{\topsep}{0pt}\setlength{\parindent}{0pt}\@skiphyperreftrue\algoflt[#1]\caption{#2}}{\endalgoflt\@skiphyperreffalse}
\makeatother

\usepackage{subcaption}
\usepackage{graphicx}
\hypersetup{hidelinks}

\newtheorem{theorem}{Theorem}
\newtheorem{lemma}[theorem]{Lemma}
\newtheorem{proposition}[theorem]{Proposition}
\newtheorem{definition}[theorem]{Definition}
\newtheorem{claim}[theorem]{Claim}

\newcommand{\pr}[1]{\ensuremath{\mathbf{Pr}\left[#1\right]}}
\newcommand{\E}[1]{\ensuremath{\mathbf{E}\left[#1\right]}}
\newcommand{\ddw}{\ensuremath{\texttt{VOG}}\xspace}
\newcommand{\dist}[2]{\ensuremath{\operatorname{dist}_G}\left(#1,#2\right)}
\newcommand{\EdgeVOG}{\texttt{EdgeVOG}\xspace}
\newcommand{\ato}[1]{\uparrow #1}
\newcommand{\Est}[1]{\underset{\scriptscriptstyle \Vec{S}(t) = \Vec{s}}{\mathbf{E}}\left[#1\right]}

\usepackage{enumitem}
\newcommand{\growingmid}{\mathrel{}\middle|\mathrel{}}
\newcommand{\tO}[1]{\ensuremath{\widetilde{O}\left(#1\right)}\xspace}

\title{\textsc{\LARGE The Voter Model with Opinion Graph}}

\usepackage{authblk}
\author[1]{Petra \textsc{Berenbrink}}
\author[1]{Thorsten \textsc{Götte}}
\author[2]{Martin \textsc{Hoefer}}
\author[2]{Tim \textsc{Koglin}}
 \affil[1]{\small
 University of Hamburg, Germany \authorcr \href{mailto:petra.berenbrink@uni-hamburg.de}{\emph{petra.berenbrink@uni-hamburg.de}} \hspace{1cm} \href{mailto:thorsten.goette@uni-hamburg.de}{\emph{thorsten.goette@uni-hamburg.de} } \authorcr \phantom{aaa}}
 \affil[2]{\small
 RWTH Aachen University, Germany \authorcr
 \href{mailto:mhoefer@cs.rwth-aachen.de}{\emph{ mhoefer@cs.rwth-aachen.de}} \hspace{1cm}
  \href{mailto:koglin@algo.rwth-aachen.de}{\emph{koglin@algo.rwth-aachen.de}}
}

\date{\today}

\usepackage{libertine}
\usepackage[scaled=0.95]{inconsolata}

\usepackage[
    backend     = bibtex,
    doi         = true,
    giveninits  = true,
    isbn        = false,
    sortcites   = true,
    style       = numeric,
    url         = false,
    maxnames = 99,
    minnames = 99,
]{biblatex}

\AtEveryBibitem{
    \clearfield{month}
    \clearlist{location}
    \clearlist{publisher}
    \clearname{editor}
}

\begin{document}

\maketitle

\thispagestyle{empty}

\begin{abstract}
We study the \emph{voter model with opinion graph}, a generalization of the standard (pull) voter model that adds a notion of bounded confidence.
In this process, we consider a finite set $A$ of $n$ agents.
Each agent $u \in A$ has an opinion $x_u \in O$, where we model opinions as vertices in an undirected opinion graph $G = (O,E)$ with $|O|=q$. In each step, an agent $u$ interacts with another agent $v \in A$.
In contrast to classical voting, $u$ adopts the opinion $x_v$ if and only if $x_v$ and $x_u$ are neighbors in $G$.
The process is stable once all remaining opinions are an independent set in $G$.
We analyze two variants of the process: the ordinary EdgeVOG process, which samples a pair of interacting agents uniformly at random from all possible pairs, and the productive EdgeVOG process, which samples only pairs with neighboring opinions in $G$.

We show the following:
\begin{enumerate}
\item For every opinion graph and every initial configuration, ordinary EdgeVOG stabilizes within $O(n^2\log n)$ steps in expectation and with high probability. The expectation bound is asymptotically tight. For the productive EdgeVOG process, the expected stabilization time is exactly the initial number of compatible agent pairs, while $O(n^2\log n)$ steps suffice with high probability.
\item  We show that an opinion $i$ with initial support $s_i$ survives with probability at least
$p_i:=s_i/(s_i+\sum_{j\in N(i)}s_j)$.
Further, let $\mu = \sum_{i \in O} p_i$, then we show the probability that less than $(1-\delta)\mu$ opinions survive is at most $\exp(-\delta^2\mu/(2(\Delta^2+1)))$ for every $0\leq\delta\leq1$.
\item If $G$ has maximum degree $\Delta$ and every opinion initially has support at most $\hat{s}$, then no opinion ever acquires more than $\tilde{O}(\min\{n,\Delta^6\hat{s}\})$ agents with high probability. This yields near-linear stabilization bounds for productive EdgeVOG on bounded-degree opinion graphs with well-spread initial configurations.
\end{enumerate}

\end{abstract}

\vfill

\paragraph{Acknowledgements}  We used ChatGPT 5.6 Sol and 6 Astra to assist with improving clarity, typesetting in \LaTeX, finding related work, as well as checking and simplifying the proofs. The authors verified the correctness and originality of all content, including references.

\clearpage
\thispagestyle{empty}
\tableofcontents

\clearpage
\setcounter{page}{1}

\section{
Introduction}

\emph{Opinion dynamics} are a model to study how the views held by a population evolve through social interaction. They have been investigated extensively in physics, sociology, and artificial intelligence, both analytically and experimentally~\cite{Lorenz07,noorazar2020recent,flache2017models,sobkowicz2009modelling}. Beyond their role in modeling social behavior, opinion dynamics also provide a rich setting for the design and analysis of distributed information spreading algorithms, a central topic in theoretical computer science and applied mathematics.
Models originating in physics and sociology typically represent opinions as continuous quantities and often assume that agents update their views by averaging the opinions of their neighbors. Prominent examples include the French--DeGroot~\cite{Degroot74,French1956}, Friedkin--Johnsen~\cite{FriedkinJ90}, Hegselmann--Krause~\cite{HegselmannK02}, and Deffuant--Weisbuch~\cite{DeffuantNAW00} models. By contrast, models studied in distributed computing and applied mathematics usually assume discrete opinions and randomized pairwise interactions. Examples include the standard voter model~\cite{BerenbrinkGKM16,CooperEOR13,HassinP01,NakataIY00}, $k$-majority dynamics~\cite{BecchettiCNPST14,BerenbrinkCEKMN17,GhaffariL18}, and several generalized or hybrid variants~\cite{CooperR16,doi:10.1137/16M1105979,DBLP:conf/icalp/CooperDFR16,DBLP:conf/waw/CooperKR23,DBLP:conf/waw/CooperKRV24}.
A central feature of many continuous models, but one that is largely absent from discrete opinion dynamics, is the notion of so-called \emph{bounded confidence}. Under this principle, agents are influenced only by opinions that are sufficiently close to their own and disregard distant opinions. This assumption reflects many real social interactions, in which not all viewpoints are regarded as equally persuasive. Voters, for example, may discount arguments associated with opposing ideological camps, while citizens may reconsider a policy position only when it is presented within an acceptable range of their prior beliefs~\cite{enelow1984spatial}.

In this paper, we study a discrete pull-voter process with bounded confidence. There are $n$ agents and a set $O$ of $q$ possible opinions. The opinions form the vertex set of an undirected \emph{opinion graph} $G=(O,E)$, whose edges encode which opinions are compatible. At each step, we pick a pair of agents $(u,v)$ uniformly at random. If the current opinions of $u$ and $v$ are adjacent in $G$, then $u$ adopts $v$'s opinion; otherwise, the configuration remains unchanged. We refer to this process as the voter model with opinion graph. We consider two variants of the process. In the \emph{ordinary} \EdgeVOG process, the ordered pair is sampled uniformly from all pairs of distinct agents. This process may contain many idle steps, since agents with equal or incompatible opinions do not cause an update. We therefore also consider the \emph{productive} \EdgeVOG process, in which the interaction pair is sampled uniformly from the pairs whose current opinions are adjacent in $G$. We call these the productive steps.

Our goal is to understand how the structure of the opinion graph affects the behavior of the process. In particular, we study how long the process takes to stabilize, how many opinions survive, and how large the support of an individual opinion can become. When the opinion graph is complete, our model reduces to the classical voter model. In this case, the expected stabilization time is $O(n^2)$, exactly one opinion survives and therefore eventually has support $n$, and a fixed opinion $i$ survives with probability $s_i/n$, where $s_i$ is its initial support.
These properties can change substantially when the opinion graph is sparse and the agents are well spread over the opinions\footnote{Note that one can always place $\Theta(n)$ agents on two neighboring opinions to enforce a large stabilization time and a large maximal support. Thus, the impact of the opinion  graph $G$ is arguably most significant if the agents are uniformly distributed over the opinions.}. Each opinion can then exchange agents with only a few other opinions. Under a uniform initial distribution, there are consequently fewer compatible pairs of agents, and an opinion becomes stable as soon as its few neighboring opinions disappear. One would therefore expect sparse opinion graphs to require fewer productive steps, preserve more opinions, and prevent individual opinions from acquiring very large support. Our results make these intuitions precise and quantify all three effects in terms of the opinion graph and the initial distribution.

\begin{enumerate}[noitemsep,topsep=0pt,leftmargin=\parindent,align=left,labelwidth=\parindent,labelsep=0pt]
\item \emph{Stabilization time.} For the productive \EdgeVOG process, we show that the expected stabilization time is \textbf{exactly} the initial number of pairs of agents whose opinions are adjacent. In particular, for a fixed distribution of agents to opinions, deleting edges from the opinion graph can only decrease the expected number of productive steps.
Surprisingly, however, there are instances with only $n-1$ such pairs that take $\Omega(n^2)$ steps with probability $\Omega(\tfrac{1}{n})$. Thus, the stabilization time does not need to be concentrated around its expectation. To obtain high-probability bounds despite this heavy tail, we show that the stabilization time is instead controlled by the largest support attained by any opinion. This connection allows us to derive refined stabilization bounds from our bounds on the maximum support and the graph's degree. For the ordinary process, which includes idle steps, we show that stabilization takes $O(n^2\log n)$ steps in expectation and with high probability on every opinion graph. The expectation bound is asymptotically tight, even for opinion graphs of maximum degree one. Thus, while sparsity reduces the number of productive steps, it need not reduce the number of ordinary steps, since the waiting time between productive interactions may dominate.
The stabilization analysis is driven by the right choice of potential: the number of pairs of agents holding adjacent opinions. We show that this quantity decreases by one in expectation whenever agents with neighboring opinions interact. Using this insight, the stabilization analysis of both processes is relatively direct.

\item \emph{Surviving opinions.} Next, we show that an opinion $i$ with initial support $s_i$ survives with probability at least
$\nicefrac{s_i}{s_i+\sum_{j\in N_i}s_j}$.
This tightens and generalizes the bound of $s_i/n$ from the classical voter model. More importantly, however, we prove that many opinions survive simultaneously when the graph is sparse. Let $\mu:=\sum_{i\in O}\nicefrac{s_i}{s_i+\sum_{j\in N_i}s_j}$ denote the sum of the individual survival guarantees. Then the probability that fewer than any fixed constant fraction of $\mu$ opinions survive is exponentially small in $\mu/\Delta^2$. For graph classes such as trees and grids, this bound can be improved even further.
The survival probability of a fixed opinion follows from observing that the ratio between an opinion's support and the total support of its closed neighborhood increases in expectation and then applying standard martingale results. Proving that many opinions survive simultaneously is significantly more difficult because the survival events are not independent. The key observation is that the extinction events of opinions at distance at least three are negatively correlated (in a specific way). This allows us to prove a variant of the Chernoff bound tailored to this dependence structure.

\item\emph{Maximum support.} Finally, we show that sparse opinion graphs prevent individual opinions from becoming dominant. If the opinion graph has maximum degree $\Delta$ and every opinion initially has support at most $\hat s$, then, with high probability, no opinion ever acquires more than $\widetilde O(\Delta^6\hat s)$ agents. In particular, for constant $\Delta$ and a uniform initial distribution, every opinion has support at most $\widetilde O(n/q)$ throughout the process. This contrasts with the classical voter model, in which one opinion eventually acquires all $n$ agents. Together with our connection between maximum support and stabilization time, this also yields a stabilization time of $\widetilde O(n^2/q)$ with high probability.
This is the most technically involved part of the paper. For the proof, we view the agents as tokens moving along the opinion graph. If an opinion acquires large support, then many tokens must have reached it along paths through other opinions. Tracing these paths backwards produces a witness structure in which many opinions must remain undecided while participating in many productive interactions. We show that large such structures are unlikely. Before giving the general argument, we illustrate its main ideas on cycles with uniform initial support, where we bound the maximum support by $O(\log n)$.
\end{enumerate}

The underlying interaction rule has appeared in earlier work. Lanchier and Scarlatos~\cite{LanchierScarlatos2017} focus on an infinite line of agents, where each agent interacts only with its two neighbors. They determine when every agent continues to change opinion indefinitely, and when instead every agent eventually stops changing. Hardin and Lanchier~\cite{HardinLanchier2022} consider arbitrary finite communication graphs, but their main results concern the probability of complete consensus, i.e., that exactly one opinion survives. Öztürk~\cite{DBLP:journals/advcs/Ozturk13} restricts the opinion graph to powers of a path; for this graph class, he derives the exact expected number of steps. In contrast, our goal is to understand systematically how the structure of an arbitrary opinion graph affects the stabilization time, the number of surviving opinions, and the support attained by individual opinions.

\paragraph{Structure of This Paper.}
In the remainder of this section we first present the related work.
Then, in \Cref{sec:preliminaries}, we formally introduce the required notation and preliminaries.
We then formally state the process, main results, and a sketch of our proof techniques in \Cref{sec:results}.
In \Cref{sec:proofs}, we present the full proofs of our theorems.
Finally, we present tools and folklore we use in the analysis (\Cref{apx:tools}) and show how our techniques can be adapted to other variants of the problem (\Cref{apx:voting}).
 \newpage
\subsection{Related Work}
\label{sec:relatedwork}

\paragraph{The Standard Voter Model:}
\begin{wrapfigure}{r}{0.4\textwidth}
\centering
    \begin{tikzpicture}[scale=0.73]

    \tikzset{
        grid node/.style={
            circle,
            draw,
            thick,
            fill=white,
            minimum size=5.5mm,
            inner sep=0pt
        },
        local node/.style={
            grid node,
            fill=blue!18
        },
        center node/.style={
            grid node,
            fill=blue!70
        },
        grid edge/.style={
            thick,
            gray!65
        },
        neighborhood/.style={
            draw=blue!70!black,
            thick,
            rounded corners=3pt,
            fill=blue!8
        }
    }

    \draw[neighborhood]
        (-0.45,2)
        -- (2,-0.45)
        -- (4.45,2)
        -- (2,4.45)
        -- cycle;

    \foreach \x in {0,...,3} {
        \foreach \y in {0,...,4} {
            \pgfmathtruncatemacro{\nextx}{\x+1}
            \draw[grid edge] (\x,\y) -- (\nextx,\y);
        }
    }

    \foreach \x in {0,...,4} {
        \foreach \y in {0,...,3} {
            \pgfmathtruncatemacro{\nexty}{\y+1}
            \draw[grid edge] (\x,\y) -- (\x,\nexty);
        }
    }

    \foreach \x in {0,...,4} {
        \foreach \y in {0,...,4} {
            \pgfmathtruncatemacro{\dist}
                {abs(\x-2)+abs(\y-2)}

            \ifnum\dist=0
                \node[center node] at (\x,\y) {};
            \else
                \ifnum\dist<3
                    \node[local node] at (\x,\y) {};
                \else
                    \node[grid node] at (\x,\y) {};
                \fi
            \fi
        }
    }

\end{tikzpicture}
\caption{Opinion graph for the two-dimensional DW process with $\tau=2$, i.e., a grid graph with connections to all nodes in distance $2$ (light blue). }
\end{wrapfigure}
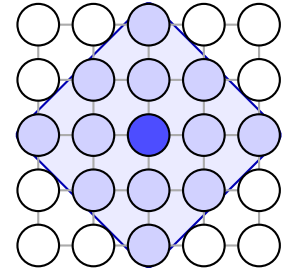
Our process can be regarded as a generalization of the well-known discrete voter (pull) process.
The voter process was first rigorously studied in~\cite{Welsh} in a continuous time model.
In~\cite{DBLP:conf/cocoon/NakataIY99,HassinP01}, the voter process is studied in a discrete-time and parallel setting, where in each round every agent adopts the opinion of a randomly chosen neighbor.
Initially, there are two opinions $A$ and $B$, and it is shown that all agents eventually agree on one of the two opinions.
In \cite{HassinP01} the convergence time of the voter process with two initial opinions is expressed in terms of the meeting time of two independent discrete-time random walks, yielding a worst-case bound of $O(n^3\log n)$.
In addition, there are works that tie the convergence time directly to graph parameters.
\cite{CooperEOR13} considers an arbitrary number of opinions and derive an upper bound of $\tilde{O}(\tfrac{1}{(1-\lambda_2)}\cdot\tfrac{n}{\nu})$, where the parameter $\nu$ measures the regularity of the degree sequence, ranging from $1$ for regular graphs to $n$ for the star graph.
\cite{BerenbrinkGKM16} proves an upper bound of $O(m/(d_{\min}\cdot \Phi))$, where $m$ is the number of edges, $d_{\min}$ is the minimal degree, and $\Phi$ the conductance of the graph.

\paragraph{The Deffuant-Weisbuch Model:}
Moreover, our process is related to continuous opinion dynamics with bounded confidence.
In particular, our model is partly inspired by the Deffuant-Weisbuch (DW) process~\cite{DeffuantNAW00,WeisbuchDAN03} with bounded confidence.
In this model, we consider a set $A$ of $n$ agents, where each agent $u \in A$ holds an opinion $x_u(t) \in [0,1]$ at time $t \geq 0$.
In each step, an unordered pair of agents $\{u,v\}$ is selected uniformly at random.
If $|x_u(t)-x_v(t)|<\varepsilon$ for a fixed confidence bound $\varepsilon>0$, both agents move a $\mu$-fraction towards each other, for a convergence parameter $\mu\in(0,1/2]$. Otherwise, no update occurs.
Depending on the parameters $\varepsilon$, $\mu$, and the initial configuration, the bounded-confidence mechanism can lead to consensus, clustering, or persistent polarization.
Our model can be seen as a discretization of the DW process where we replace the standard averaging behavior by a random pull operation, as discussed in the introduction.
Informally speaking, our update rule removes the notion of \emph{compromise} and instead lets one of the two agents completely adopt the other's opinion.
While the DW process has been extremely popular over the last decades, the vast majority of works rely on experimental analysis~\cite{Lorenz07,Lorenz10}.
Analytical results on convergence properties include bounds for asymmetric~\cite{ChenZ15} and heterogeneous~\cite{ChenSMB20} variants.
More broadly, many bounded-confidence models for opinion formation have been studied; a recent survey summarizes mathematical properties and provides further references~\cite{BernardoAPV24}.

\paragraph{The Axelrod Model:}
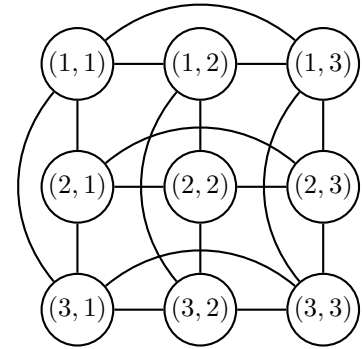
\begin{wrapfigure}{R}{0.4\textwidth}
\centering
\tikzset{
    opinion/.style={
        circle,
        draw,
        thick,
        fill=white,
        minimum size=7.5mm,
        inner sep=1pt,
        font=\small
    },
    opinion edge/.style={thick}
}

\begin{tikzpicture}[x=1.25cm,y=1.25cm]

    \foreach \i in {1,2,3} {
        \foreach \j in {1,2,3} {
            \pgfmathsetmacro{\y}{3-\i}
            \pgfmathsetmacro{\x}{\j-1}
            \node[opinion] (x\i\j) at (1.3*\x,1.3*\y)
                {$(\i,\j)$};
        }
    }

    \foreach \i in {1,2,3} {
        \draw[opinion edge]
            (x\i1) -- (x\i2) -- (x\i3);

        \draw[opinion edge]
            (x\i1) to[bend left=40] (x\i3);
    }

    \foreach \j in {1,2,3} {
        \draw[opinion edge]
            (x1\j) -- (x2\j) -- (x3\j);

        \draw[opinion edge]
            (x1\j) to[bend right=40] (x3\j);
    }

\end{tikzpicture}

\caption{The opinion graph of the Axelrod model with $F=2$ and $q=3$.}
\label{fig:ozturk-axelrod-opinion-graphs}
\end{wrapfigure}
The Axelrod model is a simple way of exploring how people influence one another. It shows how local interactions can make communities more alike while still allowing distinct cultural groups to emerge.
Seeking formalization, let $G=(V,E)$ be an undirected graph with $n=|V|$ and $m=|E|$.
Each vertex $v$ has a so-called culture $X_v=(X_{v,1},\ldots,X_{v,F})\in[q]^F$.
For $\{u,v\}\in E$, let the difference between the cultures of the two endpoints be  $s(u,v)=|\{i:X_{u,i}=X_{v,i}\}|$.
The process now updates the cultures as follows: At each step, choose an edge uniformly.
With probability $s(u,v)/F$, choose a disagreeing feature uniformly
and let a uniformly chosen endpoint copy the other's trait.
A configuration is absorbing iff for all edges, it holds
$s(u,v)\in\{0,F\}$, i.e.,
neighboring agents are either identical or completely
incompatible. Consequently, the process may end in consensus or in
several mutually incompatible cultural regions.
In the case $F=2$, i.e., when each agent has precisely two traits (and each trait can take $q$ different values), the Axelrod model coincides with our model on the $2$-dimensional rook graph on $q^2$ nodes (if we exclude the idle steps).

\paragraph{Lotka Volterra Processes:}
Finally, a related framework is studied in~\cite{CZYZOWICZ20221}.
The authors consider a broad class of protocols with randomly interacting agents whose opinion updates follow so-called discrete Lotka-Volterra (LV-type) dynamics.
More precisely, in each step, a directed pair of agents $(v,w)$ is picked uniformly and independently at random.
However, in this framework, the opinion graph can be \emph{directed and weighted}: For any two opinions $i,j$ there is probability $p_{ij} \in [0,1]$, s.t., with probability $p_{ij}$ an agent with opinion $i$ switches to $j$ if it interacts with an agent with opinion $j$.
Our protocol falls within this class as we choose $p_{ij} = 1$ if there is an edge $\{i,j\}$ in the opinion graph and $0$ otherwise.
Their analysis establishes polynomial convergence-time guarantees, showing that any LV-type protocol converges in $O(n^{c})$ steps for some constant $c \gg 2$ (provided that all $p_{ij}$'s are constant).
In contrast, we focus on a specific protocol and provide a significantly more fine-grained analysis.
In particular, we show our class of processes converges after $O(n^2 \log n)$ steps.

\subsection{Notation \& Preliminaries}
\label{sec:preliminaries}

We consider the voter model with a simple, undirected \emph{opinion graph} $G = (O,E)$.
The nodes of $G$ are the $q$ available opinions $O = \{1,\ldots,q\}$.
We let $\Delta$ be the maximum degree in $G$.

For the analysis, we may discard opinions with initial support zero, since they can never reappear. Thus, we assume w.l.o.g. that every opinion initially has positive support and $q\leq n$. We assume $n\geq2$; the case of one agent is immediate. Passing to the induced opinion graph can only decrease its maximum degree. We use $\log$ for the natural logarithm; other bases are indicated explicitly.

Moreover, there is a set $A$ of $n$ agents and
each agent $u$ has initially one of the $q$ opinions.
The opinion of agent $u$ at the beginning of step $t$ is denoted by $o_u(t) \in O$. We define $\mathcal{S}_i(t)$ as the set of agents with opinion $i$ at the beginning of step $t$ and $S_i(t) = |\mathcal{S}_i(t)|$. We refer to $S_i(t)$ as the \emph{support} of $i$ at the beginning of step $t$.
For a subset of opinions $O' \subseteq O$, let $\mathcal{S}_{O'}(t) := \{v \in A \mid o_v(t) \in O'\}$ be the set of agents with an opinion from $O'$ in step $t$. Again, we define $S_{O'}(t)=|\mathcal{S}_{O'}(t)|$.
We denote by $P(t) = \{\{u,v\} \mid \{ o_u(t), o_v(t) \} \in E\}$  the set of \emph{productive pairs} (or \emph{productive edges}). Let the pair $\lbrace u,v \rbrace$ be chosen in step $t$. Then, we call an opinion $o\in O$ \emph{active} in step $t$ if $o_u(t) = o$ or $o_v(t) = o$. Likewise, an edge $\{o, o' \} \in E$ is active (in step $t$) if $o = o_u(t)$ and $o' = o_v(t)$.
We call $\vec{S}(t)$ the configuration of the system at time $t$. Since the interactions are the same for all agents with the same opinion, we can simply define $\vec{S}(t)$ as $(S_1(t), \ldots, S_q(t))$. We denote the initial configuration by $\vec{s}_0$.
Opinion $i$ is \emph{decided} in step $t$ if either $S_i(t) = 0$ \emph{or} we have $S_j(t) = 0$ for all opinions $j$ with $\{i,j\} \in E$. Otherwise, $i$ is undecided. Note that once an opinion is decided, it will never become undecided again. The \emph{stabilization time} is the first time step in which all opinions are decided.
 \newcommand{\EX}[1]{\underset{\scriptscriptstyle X}{\mathbf{E}}\left[#1\right]}
Moreover, we use some common terms from the analysis of (randomized) algorithms. We will use the $\Tilde{O}(\cdot)$-notation, which suppresses polylogarithmic factors (similar to the $O(\cdot)$-notation, which suppresses constant factors).
We say an event holds \emph{with high probability} (w.h.p.) if it holds with probability at least $1 - \nicefrac{1}{n^c}$ for some $c > 1$.
Finally, we use sometimes use $\EX{\cdot}$ as shorthand for $\E{\cdot \growingmid X}$.

\clearpage

\section{Our Results and Techniques}
\label{sec:results}
In this section, we give an overview of our results and the main techniques used in the respective proofs.
We study two variants of the \ddw process, called the \emph{ordinary} \EdgeVOG process and the \emph{productive} \EdgeVOG process. Both variants are defined with respect to a fixed opinion graph $G=(O,E)$ and a set of agents $A$. In each step $t$, a pair of agents ${u,v}\in\binom{A}{2}$ is selected. If their current opinions are compatible, that is, if ${o_u(t),o_v(t)}\in E$, then one of the two agents adopts the opinion of the other, with each direction chosen with probability $1/2$. More precisely, either $o_u(t+1)=o_v(t)$ or $o_v(t+1)=o_u(t)$, while all other opinions remain unchanged.
If the two opinions are not adjacent in $G$, no update occurs.
The two variants differ only in how the interacting pair is selected. In the ordinary \EdgeVOG process, the pair is chosen uniformly at random from all pairs in $\binom{A}{2}$. In the productive \EdgeVOG process, it is chosen uniformly from the currently productive pairs, namely those pairs whose opinions are adjacent in $G$.

Note that the ordinary process models a well-mixed population: the communication graph is complete, and every pair of agents is equally likely to interact. This standard mean-field assumption is also used in classical versions of the voter model \cite{HassinP01}, the Deffuant–Weisbuch model \cite{WeisbuchDAN03}, the Hegselmann–Krause model \cite{HegselmannK02}, and discrete Lotka-Volterra processes \cite{CZYZOWICZ20221}. The authors of \cite{CZYZOWICZ20221} note that these are the same interaction rules as in a population protocol. It is easy to see that, after removing all idle steps, the ordinary process has exactly the same distribution as the productive process.

\noindent \begin{minipage}[b]{0.49\textwidth}
\begin{algorithm}{\;Ordinary \EdgeVOG on $G = (O,E)$}
In each step $t = 0,1,2, \ldots$:
    Pick a pair $\{u,v\}\in\binom{A}{2}$ u.a.r.
    if $\{o_u(t), o_v(t)\} \in E$:
        Flip a fair coin $c \in \{\textsc{Heads},\textsc{Tails}\}$
        if $c = \textsc{Heads}$ $o_u(t+1) \gets o_v(t)$;
                else $\,\,o_v(t+1) \gets o_u(t)$
                $$
\end{algorithm}
    \end{minipage}
\quad
    \begin{minipage}[b]{0.49\textwidth}
\begin{algorithm}{\; Productive \EdgeVOG on $G = (O,E)$}
In each step $t = 0,1,2, \ldots$:
    Let $P(t) \coloneqq \{\{u,v\} \in \binom{A}{2} \mid \{o_u(t),o_v(t)\} \in E \}$
    if $P(t)=\emptyset$: stop
    Pick a pair $\{u,v\}\in P(t)$ u.a.r.
    Flip a fair coin $c \in \{\textsc{Heads},\textsc{Tails}\}$
    if $c = \textsc{Heads}$ $o_u(t+1) \gets o_v(t)$;
            else  $\,\, o_v(t+1) \gets o_u(t)$
\end{algorithm}
\end{minipage}

Our results are organized as follows. In \Cref{sec:worstcase} we study the convergence times of the ordinary and of the productive \EdgeVOG processes. In \cref{sec:distribution}, we consider the properties of the stable configuration.
We give a lower bound on the probability that a given opinion survives and a tail estimate on the number of surviving opinions.
We then bound the maximum support attained by any opinion during the process in \cref{sec:maxsupport}.
Finally, we summarize our results in \cref{sec:conclusion}.

\subsection{Bounds on the Stabilization Time}
\label{sec:worstcase}

We begin by presenting our potential function $\Psi$, which is the main tool of our analysis.
Then, we present our worst-case bounds for the ordinary and productive process that are proved with the help of this potential.
Finally, we give a parameterized bound for the productive process that depends on the opinion graph degree and the maximal support.

\paragraph{The Potential Function $\Psi$.}
For the analysis of both our processes, we use a (family of) potential function(s): For an edge ${i,j}\in E$, we define its potential at time $t$ as $\Psi_{ij}(\Vec{S}(t))=S_i(t)\cdot S_j(t)$, where $S_i(t)$ denotes the support of opinion $i$. The potential associated with an opinion $i\in O$ is
\begin{align*}
\Psi_i(\Vec{S}(t))
=\sum_{j\in N_i}\Psi_{ij}(\Vec{S}(t)) && \textbf{(Potential of opinion $i$ in step $t$)}.
\end{align*}
An opinion $i$ is called \emph{decided} once either $S_i(t)=0$ or all its neighboring opinions have support zero. Equivalently, $i$ is decided precisely when $\Psi_i(\Vec{S}(t))=0$.
Finally, we define the potential of the entire configuration by
\begin{align*}
\Psi(\Vec{S}(t)) =\sum_{{i,j}\in E}\Psi_{ij}(\Vec{S}(t)) = \frac{1}{2} \sum_{i \in O}\Psi_{i}(\Vec{S}(t))  && \textbf{(Potential of $\Vec{S}(t)$)} .
\end{align*}
Note that $\Psi(\vec S(t))=|P(t)|$, since both quantities count the productive pairs. Hence, we can define the stabilization time as $T:=\inf\{t\geq 0:\Psi(\Vec{S}(t))=0\}.$
If clear from context, we will simply write $\Psi(t)$ instead of $\Psi(\Vec{S}(t))$.
In \cref{sec:potential}, we establish several properties of these potential functions that will be used throughout the paper.
Crucially, we will show that this potential exhibits a negative drift in expectation for both the ordinary and productive processes.
\begin{restatable}{prop}{newpotential}
\label{lemma:new_potential}
Let $\mathcal{P}(t) \in \{0,1\}$ indicate that step $t$ is productive. Then, for every possible configuration $\Vec{s}$ it holds
$$
    \Est{\Psi(\Vec{S}(t+1))} = \Psi(\Vec{s}) -  \Est{\mathcal{P}(t) \xspace} .
$$
\end{restatable}
\begin{proof}[Proof sketch]
The statement follows from straightforward calculation and the law of total expectation.
Condition on $\{\Vec{S}(t)=\vec{s}\}$.
If step $t$ is not productive, the potential remains unchanged.
If the step is productive and involves opinions $i$ and $j$, then one unit of support is transferred between $i$ and $j$, with either direction being equally likely.
With a straightforward calculation, it can be shown that $\Psi_{ij}$ decreases by one in expectation as
\begin{align*}
    \frac{1}{2}(s_i-1)(s_j+1)+\frac{1}{2}(s_i+1)(s_j-1)=s_i s_j-1.
\end{align*}
All other contributions remain unchanged in expectation: Consider an edge $\{i,k\}$ with $k \neq j$. As $s_i$ increases or decreases with equal probability and $s_k$ does not change, the potential $\Psi_{ik}$ does not change in expectation. The same holds for $\Psi_{jk}$ with $k \neq i$.
 The potentials of edges disjoint from $\{i,j\}$ do not change at all.
 Thus, every productive step decreases $\Psi$ by exactly one in expectation, independently of the chosen productive edge. The statement then follows by the law of total expectation. The full proof with all calculations can be found in \cref{sec:potential}
\end{proof}
This innocent-looking proposition is the key to many of our results, as it
directly relates the convergence time to the number of productive steps.
As we don't restrict how edges are chosen, the proposition is quite flexible and may also be of independent interest.
For any rule that decides on which pairs are picked, we simply need to bound the time in which $\Psi(0)$ productive pairs are picked to obtain the expected runtime.
In particular, it gives a simpler proof of the corresponding bound for asynchronous pull voting on regular communication graphs in \cite{CooperR16}; see \cref{apx:voting}.
Note that the potential has been used before by Öztürk\cite{DBLP:journals/advcs/Ozturk13}.

\paragraph{Stabilization Time of the Ordinary Process}
Given our potential, we consider the worst-case convergence times of the ordinary \EdgeVOG process.
Rather surprisingly, the bounds for the ordinary process do not depend on the initial configuration nor the opinion graph.
Instead, regardless of the opinion graph and the initial configuration, the convergence time is at most essentially quadratic in the number of agents, namely
\begin{restatable}{paperthm}{newPotential}
\label{thm:newPotential}
Let $G = (O,E)$ be an arbitrary opinion graph and suppose the process
starts in an arbitrary configuration.
Then, the ordinary \EdgeVOG process converges in ${O}(n^2\log n)$ steps in expectation and w.h.p. The bound in expectation is asymptotically tight, even when $\Delta=1$.
\end{restatable}
\begin{proof}[Proof sketch]
For the proof, we consider the expected potential drop from step to step.
Suppose that the process has potential $\psi$ at step $t$. As $\psi$ is equal to the number of productive pairs and in a complete communication graph each pair is picked with probability ${1}/{\tbinom{n}{2}}$, the probability that a step is productive is ${\psi}/{\tbinom{n}{2}}$.
Therefore, by \cref{lemma:new_potential}, in each step the expected drop in potential is ${\psi}/{\tbinom{n}{2}}$ as well.
Thus, by the law of total expectation, we get

\begin{align*}
    \E{\Psi(\Vec{S}(t+1))}
    &=\E{\Psi(\Vec{S}(t))-\frac{\Psi(\Vec{S}(t))}{\binom n2}} =\left(1-\frac{1}{\binom n2}\right)\E{\Psi(\Vec{S}(t))}.
\end{align*}

This fact alone is sufficient to prove \cref{thm:newPotential} as we apply standard drift theorems (see \cite{Lengler2020}) to derive the stated bounds.
The full proof is given in \Cref{sec:general graph bound}.
\end{proof}
Recall that the standard voter model on a complete communication graph has a worst-case convergence time of $\tilde{\Theta}(n^2)$ steps in expectation and with high probability~\cite{CooperR16}.
Thus, our result provides two observations:
First, up to logarithmic factors, \EdgeVOG has asymptotically the same worst-case convergence time as the standard voter despite restricting the set of admissible opinions via $G$.

Second, the bound in expectation is tight, including the logarithmic factor. To see this, let $n$ be even and let the opinion graph consist of $n/2$ disjoint edges, with one agent initially holding each opinion.
Each component contains exactly two agents and stabilizes as soon as this pair is selected for the first time.
When $r$ components remain undecided, there are exactly $r$ productive pairs, so the expected waiting time until the next component stabilizes is $\binom n2/r$.
Consequently,
\[
    \E{T_{\mathrm{ordinary}}}
    =\binom n2\sum_{r=1}^{n/2}\frac1r
    =\Theta(n^2\log n).
\]
In contrast, the productive process stabilizes after exactly $n/2$ steps, since each productive interaction resolves one component.

\paragraph{Worst-Case Stabilization Time of the Productive Process}
We next consider the productive process, which omits all idle steps in which agents with non-adjacent opinions interact. Under a natural coupling, it therefore cannot take more steps than the ordinary process. However, rather surprisingly its stabilization time can still have a long tail: although its expectation is exactly the initial number of productive pairs, there are configurations with only $n-1$ productive pairs for which the process takes $\Omega(n^2)$ steps with probability $\Omega(\tfrac{1}{n})$.
To be precise, it holds:
\begin{theorem}
\label{thm:ProdWorstCase}
Let $G = (O,E)$ be an arbitrary opinion graph and suppose the process starts in a configuration $\Vec{s_0}$ with $|P(0)|$ productive pairs. Then, the productive \EdgeVOG process converges in \textbf{exactly} $|P(0)|$ steps in expectation and $O(n^2 \log n)$ w.h.p.
Furthermore, there are graphs and configurations with $\Psi(0) = n-1$ that require $\Omega(n^2)$ steps with prob. $\Omega(\tfrac{1}{n})$.
\end{theorem}
\begin{proof}[Proof Sketch]
We first determine the expectation. Since every step is productive, \cref{lemma:new_potential} shows that the potential decreases by exactly one in expectation in every step. The potential is non-negative and equals zero exactly when the process has stabilized. Thus, the additive drift theorem gives $\E{T}=\Psi(\Vec{s_0})=|P(0)|$.

For the tightness result, consider a star graph with $n$
opinions and place one agent on each opinion. Let $i$ be the central
opinion. Every productive interaction involves $i$ and one of the leaves,
and either agent adopts the opinion of the other with probability $1/2$.
Therefore, $S_i(t)$ performs a simple symmetric random walk on
$\{0,\ldots,n\}$ starting at $1$. The process stabilizes precisely when
this walk reaches $0$ or $n$.
By the standard gambler's-ruin formula, the expected absorption time is
$n-1$. Moreover, the walk reaches $\lfloor n/2\rfloor$ before reaching
$0$ with probability $\Theta(1/n)$. From there, it takes
$\Omega(n^2)$ additional steps with constant probability. Thus, the
stabilization time is $\Omega(n^2)$ with probability $\Omega(1/n)$.
The full proof is given in \cref{sec:wcproductive}.
\end{proof}

\paragraph{Parameterized Stabilization Time of the Productive Process}
Although the bounds from the previous theorem are tight for some opinion graphs and initial configurations, many natural instances stabilize much faster than the worst-case bound of $O(n^2\log n)$. Intuitively, if the opinion graph is sparse and the agents are initially spread across many opinions, the occupied part of the opinion graph should quickly break into small components, which then stabilize rapidly. This suggests that a more refined analysis should take into account both the structure of the opinion graph and the initial distribution of the agents. The following theorem confirms this intuition at a general level.

We bound the stabilization time in terms of the initial number of productive pairs, the largest support attained by any opinion, and the maximum degree of the opinion graph. More precisely, for an opinion graph $G$, an initial configuration $\Vec{s_0}$, and a constant $c>0$, define
\begin{align}
\hat{\mathcal S}_c(G,\Vec{s_0})
:=
\min\left\{
x\geq 0 :
\pr{
\max_{i\in O}\max_{0\leq t\leq T}S_i(t)>x
\growingmid \Vec{S}(0)=\Vec{s_0}
}
\leq n^{-c}
\right\}.
\label{eq:maxsupportboundwhp}
\end{align}
Thus, $\hat{\mathcal S}_c(G,\Vec{s_0})$ is the smallest threshold such that, with probability at least $1-n^{-c}$, no opinion ever has support larger than this threshold.

The following theorem provides a general parameterization rather than an explicit bound for a particular class of opinion graphs. In particular, it shows that the stabilization time can be subquadratic whenever the initial potential, the maximal support, and the maximum degree are sufficiently small. Since $\hat{\mathcal S}_c(G,\Vec{s_0})$ depends on both the opinion graph and the initial configuration, we derive an explicit bound on this quantity later in \cref{sec:maxsupport}.

\begin{restatable}{paperthm}{newPotentialProd}
\label{thm:newPotentialProd}
Let $G=(O,E)$ be an opinion graph with maximum degree $\Delta$, and suppose that the process starts in a configuration $\Vec{s_0}$ with $|P(0)|$ productive pairs. Furthermore, let $\hat{\mathcal S}_c(G,\Vec{s_0})$ be defined as in \eqref{eq:maxsupportboundwhp} for a large enough $c \geq 1$. Then, with high probability, the productive \EdgeVOG process stabilizes within
    $
\widetilde O\left(
|P(0)| + \Delta^2\hat{\mathcal S}_c(G,\Vec{s_0})^2 \right)
$ steps.
Alternatively, the polylogarithmic factor multiplying $|P(0)|$ can be removed at the cost of a larger polynomial dependence on $\Delta$, giving
$$
O\left(|P(0)|\right)
+
\widetilde O\left(
\Delta^6\hat{\mathcal S}_c(G,\Vec{s_0})^2
\right).
$$
\end{restatable}
\begin{proof}[Proof sketch]
Recall that, by \cref{lemma:new_potential}, the potential decreases by one in expectation in every step. However, it can still increase in individual steps. To obtain a high-probability bound, we therefore need to control how far the potential can rise before the process stabilizes. For this, our upper bound $\hat{\mathcal{S}}_c(G,\Vec{s_0})$ comes into play.

For readability, write $\hat{\mathcal S}:=\hat{\mathcal S}_c(G,\Vec{s_0})$ and $b:=\Delta\hat{\mathcal S}$ in the remainder. If the initial configuration is already stable, there is nothing to prove.
First, let $\tau_{\mathrm{supp}}$ be the first step in which some opinion has support larger than $\hat{\mathcal S}$, and let $\tilde{T}:=T\wedge\tau_{\mathrm{supp}}$.
By definition of $\hat{\mathcal S}$, the probability of such a support violation is at most $n^{-c}$. Thus, it suffices to bound $\tilde{T}$ and afterwards add $n^{-c}$ to the failure probability.
We stop the potential at $\tau_{\mathrm{supp}}$. Its stepwise changes have absolute value at most $2b$. Thus, the stopped potential is a supermartingale with bounded differences.

Suppose we start from a configuration with potential $\psi$ before a support violation. For a sufficiently large constant $C$, let
\[
    B:=C\left(\psi+b^2\log^2n\right).
\]
The maximal version of Azuma's inequality shows that the stopped potential stays below $B$ for a suitably chosen interval of $\Theta(B\log n)$ steps, w.h.p.; see \cref{lemma:azumapotential}.
We also use auxiliary node potentials $\hat\Psi_i^{(B)}$ that agree with $\Psi_i$ until the first support violation or the first time $\tau_B$ when the potential exceeds $B$, and are zero afterwards.
For opinions at pairwise distance at least four, one interaction can affect at most one node potential. In \cref{lemma:prodpotential}, we use this observation to derive a joint moment bound for the auxiliary potentials.

Equipped with these observations, we distinguish between high and low potentials. We say that a potential is high if it is larger than $H:=32c\Delta^4b^2\log^2n$; otherwise, it is low.
If the potential is hight, we partition the opinions into $4$-separated sets and apply the variant of the Chernoff bound from \cref{thm:scheidelbound} to each set. This variant applies to the specific dependencies captured by our joint moment bound.
Starting from potential $\psi$, we end a phase when the potential first falls to at most $\max\{\psi/2,H\}$ or the support bound is violated.
In \cref{lemma:inductionstep}, we show that such a phase lasts $O(\psi+b^2\log^2n)$ steps, w.h.p.
This estimate holds conditional on the history at the beginning of each phase. A union bound over $O(\log n)$ phases and summing their lengths shows that the potential falls to $H$ within $O(|P(0)|+b^2\log^3n)$ steps, w.h.p., unless the process has already stopped.

Once the potential is at most $H$, we use the joint moment bound only for singleton sets. Summing over all opinions gives, for the auxiliary global potential $\hat\Psi^{(B)}=\frac12\sum_i\hat\Psi_i^{(B)}$,
\[
    \E{\hat\Psi^{(B)}(t)}\leq\psi\left(1-\frac1B\right)^t\leq n^2e^{-t/B},
\]
where we restart the auxiliary potentials at the beginning of this part and write $\psi\leq H$ for its initial potential.
For $t=\lceil(c+4)B\log n\rceil$, this expectation is at most $n^{-(c+2)}$.
Unless one of the bounds has been violated, the auxiliary potential is a positive integer whenever the process has not stabilized. Thus, Markov's inequality, together with the Azuma bound, shows that the process stabilizes or violates the support bound within another $O(\Delta^4b^2\log^3n)$ steps, w.h.p.; see \cref{prop:smallpotential}.
Adding the probability of a support violation gives the second bound,
\[
    T\in O\left(|P(0)|+\Delta^4b^2\log^3n\right)
    =O\left(|P(0)|+\Delta^6\hat{\mathcal S}^{\,2}\log^3n\right),
\]
with high probability.
Finally, the same expectation calculation can be applied directly from the initial configuration, without first reducing the potential to $H$. This gives $O((|P(0)|+b^2\log^2n)\log n)$ steps w.h.p., proving the first bound.
The details are given in \cref{sec:whpproductive}.

\end{proof}

\subsection{Distribution of Surviving Opinions}
\label{sec:distribution}
Next, we move away from stabilization time and consider the properties of a stable configuration.
Recall that in the standard voter model, the process eventually reaches consensus, i.e.,\textbf{exactly one} opinion survives.
In our process, by contrast, multiple opinions may coexist indefinitely, provided that the surviving opinions form an independent set in the opinion graph.
In this section, we analyze how large the number of surviving opinions can be.
To this end, we first consider the survival probability of a single opinion, and then we show that many opinions survive if the opinion graph is sufficiently sparse.

\paragraph{Survival Probability of a Fixed Opinion}
Note that for a complete opinion graph it is known from the analysis of the standard pull voter process that (on a regular communication graph) an opinion survives with probability exactly $\frac{s_i}{n}$ (where $s_i$ is the initial support).
As it turns out, this bound naturally generalizes to our model.
It holds
\begin{restatable}{paperthm}{survival}
\label{thm:survival}
Let $G=(O,E)$ be the opinion graph and let $s_i\geq 1$ denote the
initial support of opinion $i\in O$.
Let
$X_i:=\mathbf 1_{\{S_i(T)>0\}}$ indicate whether opinion $i$ is present in
the stable configuration, then
$$\E{X_i}\geq  \frac{s_i}{s_i+\sum_{j\in N_i}s_j}.$$
\end{restatable}
\begin{proof}[Proof sketch]
The lemma follows from a straightforward martingale argument.
Let $M_i(t):=\sum_{j\in N_i}S_j(t)$ be the total support of the
opinions adjacent to $i$.
Note that the potential of $i$ is precisely $\Psi_i(t) = S_i(t) \cdot M_i(t)$.
Define
$\tau_i:= \inf\{t\geq 0:\Psi_i(t)=0\} = \inf\{t\geq 0:S_i(t)=0\text{ or }M_i(t)=0\}$.
Thus, $\tau_i$ is the first time at which we know if  $i$ itself is extinct or if all neighboring opinions are extinct.
Using this definition, consider the stopped process
\begin{align}
    R_i(t)
    :=
    \begin{cases}
        \displaystyle
        \frac{S_i(t)}{S_i(t)+M_i(t)}
        & \textbf{if } t<\tau_i,\\[1.2ex]
        \displaystyle
        \frac{S_i(\tau_i)}
             {S_i(\tau_i)+M_i(\tau_i)}
        & \textbf{if } t\geq\tau_i.
    \end{cases}
\end{align}
At time $\tau_i$, either $S_i(\tau_i)=0$ or $M_i(\tau_i)=0$.
Consequently, $R_i(\tau_i)=0$ if opinion $i$ becomes extinct and
$R_i(\tau_i)=1$ if all its neighboring opinions become extinct.
Since extinct opinions cannot reappear, we have
$R_i(T)=R_i(\tau_i)=X_i$.

A direct one-step calculation shows that $(R_i(t))_{t\geq 0}$ is a
bounded submartingale. An interaction between opinion $i$ and one of
its neighbors changes $S_i(t)$ symmetrically by $\pm1$ while leaving
$S_i(t)+M_i(t)$ unchanged. An interaction between a neighboring
opinion and an opinion outside $\{i\}\cup N_i$ changes the denominator
symmetrically by $\pm1$. With some straightforward calculation, we can show that $R_i(t)$ increases in expectation in this case.
All other interactions leave $R_i(t)$
unchanged. Hence, the Optional Stopping Theorem (cf. \cref{thm:doob}) yields
\begin{align}
    \E{X_i}
    =
    \E{R_i(T)}
    \geq
    R_i(0)
    =
    \frac{s_i}{s_i+\sum_{j\in N_i}s_j}.
\end{align}
For the detailed proof, see \cref{sec:survival}.
\end{proof}
Now let $Z:=\sum_{i\in O}X_i$
be the number of surviving opinions.
In the special case in which every opinion has equal support, we have $\E{Z}\geq \sum_{i\in O}\tfrac{1}{d_i+1}$, where $d_i$ is the degree of opinion $i$.
For comparison, the Caro-Wei bound guarantees that every graph with degree sequence $(d_i)_{i\in O}$ contains a maximal independent set of size at least $\sum_{i\in O} \tfrac1 {d_i+1}$.
Thus, our stochastic process produces, in expectation, an independent set whose guaranteed size matches the Caro-Wei bound.
In other words, in expectation, we sample a \emph{big} independent set.

\paragraph{Number of Surviving Opinions}

\begin{wrapfigure}{R}{0.4\textwidth}
    \centering
    \begin{tabular}{lcc}
        \hline
        Graph class
        & Arboricity
        & $\chi(G^2)$ \\
        \hline
        Forests
        & $1$
        & $\Delta+1$ \\
        Planar graphs
        & $3$
        & $O(\Delta)$ \\
        Treewidth $k$
        &  $k$
        & $O_k(\Delta)$ \\
        Euler genus $g$
        & $O(\sqrt{g}+1)$
        & $O_g(\Delta)$ \\
        $K_r$-minor-free
        & $O(r\sqrt{\log r})$
        & $O_r(\Delta)$ \\
        General graphs
        & $\lceil(\Delta+1)/2\rceil$
        & $\Delta^2+1$ \\
        \hline
    \end{tabular}
    \caption{Arboricity and chromatic number of the square for common
    graph classes. $\Delta$ is the maximal degree. The parameters $k$, $g$, and $r$ are considered fixed.}
    \label{tab:arboricity-square-coloring}
\end{wrapfigure}
Next, we investigate what we can show beyond the expectation, i.e., what is the probability that the surviving set is indeed large? To put it more technically, we want to bound the probability that $Z$ deviates from its expected size.
Unfortunately, we cannot use the standard Chernoff and Hoeffding bounds for $Z$ as the survivals of neighboring opinions are highly correlated: Two neighboring opinions cannot both survive. Furthermore, if two opinions have shared neighbors, these neighbors' supports affect both nodes' $R_i(t)$'s, the variable by which we bound the survival probability in the proof of \cref{thm:survival}.
Nevertheless, Theorem~\ref{thm:survival-tail} shows that the number of
surviving opinions is unlikely to be much smaller than $\mu:=\sum_{i\in O}p_i$ with $p_i = \nicefrac{s_i}{s_i+\sum_{j\in N_i}s_j}$, even though the individual survival events are \textbf{not} independent.
The loss compared with the usual multiplicative Chernoff bound is
entirely captured by $\kappa=\chi(G^2)$, i.e., chromatic number of
the squared opinion graph $G^2$.  Recall that two vertices are adjacent in $G^2$ if their
distance in $G$ is at most two.
More precisely, it holds
\begin{theorem}
\label{thm:survival-tail}
Let $G=(O,E)$ be an opinion graph and let
$\kappa:=\chi(G^2)$.  For every $0\leq \delta\leq 1$,
\begin{align}
    \pr{Z\leq (1-\delta)\mu}
    \leq
    \exp\left(-\frac{\delta^2\mu}{2\kappa}\right).
\end{align}
\end{theorem}
\begin{proof}[Proof sketch]
For every opinion $i$, we consider the ratio $R_i(t)=S_i(t)/(S_i(t)+M_i(t))$ from the previous proof.
Moreover, for every set $I$ of opinions at pairwise distance at least three, the product $\prod_{i\in I}(1-R_i(t))$ is a supermartingale.
Hence, for any such set, it holds
$$\E{\prod_{i\in I}(1-X_i)}\leq\prod_{i\in I}(1-p_i).$$
It is known that for binary random variables with this property, Chernoff and Hoeffding bounds hold too if we consider $\mu = \sum_{i \in I} p_i$ instead of the expectation. This has, for example, been shown in \cite{DBLP:journals/siamdm/SchmidtSS95}.
Just as the standard Chernoff bound, this result can be extended to variables whose dependencies can be modeled through a graph.
To get a result for all variables and not just the ones in distance $3$, we now consider a coloring of $G^2$ with $\kappa=\chi(G^2)$ colors.
Note that all nodes in the same color class are in distance at least 3 to each other as they form an independent set in $G^2$.
With some straightforward adaptation to the proof of the standard Chernoff bound, we derive a bound of $\pr{Z\leq(1-\delta)\mu}\leq\exp(-\delta^2\mu/(2\kappa))$.
For the detailed proof, see \cref{sec:survival-tail}.
\end{proof}
As $G^2$ has degree at most $\Delta^2$, the parameter $\kappa$ can always be bounded by $\Delta^2+1$.
Consequently, on a $\Delta$-regular opinion graph with equal initial supports, $\Omega(q/\Delta)$ opinions survive with probability $1-\exp(-\Omega(q/\Delta^3))$.
Moreover, the parameter $\kappa$ can be bounded particularly well for sparse
graph classes: A typical measure for sparseness is the arboricity, the number of forests a graph can be divided into. Agnarsson and Halldórsson~\cite{AgnarssonHalldorsson2003}
showed\footnote{ \cite{AgnarssonHalldorsson2003} showed the result for only planar graphs, but \cite{KiersteadYangYi2020} remarked that the proof can be extended to all graphs of bounded arboricity.} that, for every fixed $a$ and $p$, there exists a constant
$c$ such that every graph $G$ with arboricity $a$ satisfies
$\chi(G^p)\leq c\cdot 2^{p} \cdot a^{\lceil p/2\rceil}\cdot \Delta^{\lfloor p/2\rfloor}$.
Given this insight, for these graphs, the theorem improves to
\begin{align}
    \pr{Z\leq(1-\delta)\mu}
    \leq
    \exp\left(
        -\Omega_a\left(\frac{\delta^2\mu}{\Delta}\right)
    \right),
\end{align}
where the hidden constant polynomially depends on the arboricity bound $a$.
Additionally, as the average degree of these graphs is at most $2a$, with a uniform initialization, it holds $\mu \in \Omega(q/a)$.
Thus, for a constant $a \in O(1)$, a set of size $\Theta(q)$ survives with at least constant probability independent of $q$ and $\Delta$.
In \cref{tab:arboricity-square-coloring} we provide an overview of well-known graph classes and their respective bounds on the arboricity and the chromatic number.

\subsection{Maximal Support on Bounded-Degree Graphs}
\label{sec:maxsupport}
Finally, we bound the maximal support $\hat{\mathcal{S}}$ attained by any opinion during the process. Our main result in this section gives such a bound for opinion graphs of bounded degree when the initial support of every opinion is bounded. Recall that a bound on $\hat{\mathcal{S}}$ also gives a direct bound on the number of productive steps.
The main difficulty is that neighboring opinions are affected by the same interactions, so their supports are highly dependent. The basic idea is to show that the graph quickly breaks into small connected components of active opinions, each containing only a small number of agents. The support of any opinion is then at most the total number of agents in its component.
Before proving the general result, we present the argument for a cycle, where $\Delta=2$ and every opinion initially has support one. This special case already contains the main idea but avoids several additional complications needed for general bounded-degree graphs.

\paragraph{Maximal Support on a Cycle Graph}

On a cycle, we show that sufficiently many opinions become decided before any support grows too large. These opinions split the remaining active opinions into short paths containing only $O(\log n)$ agents. This gives the following bound.

\begin{restatable}{paperthm}{cycleshatter}
\label{prop:cycleshatter}
Let $G=(O,E)$ be a cycle on $n$ opinions, and suppose that every opinion initially has support one. Then, for both the \EdgeVOG and the productive \EdgeVOG processes, with high probability every opinion has support $O(\log n)$ throughout the process.
\end{restatable}

\begin{proof}[Proof sketch]
For the analysis we adopt a random-walk-based view of the process. Each agent can be seen as a token that is assigned to the node of the opinion graph that corresponds to its opinion. Note that, whenever an agent with opinion $i$ adopts another opinion $j$, the corresponding token moves to opinion $j$.
Thus, these tokens perform a (biased) random walk.
 We say that an opinion is decided if either it contains no tokens, or none of its neighboring opinions contains any tokens. Moreover, once an opinion becomes empty, it can effectively be removed from the graph, since no agent can ever adopt this opinion again.

With this abstraction, it is relatively straightforward to bound the maximal support for the cycle graph.
We show with high probability the process quickly decomposes into connected components containing only $O(\log n)$ agents.

To this end, fix an opinion $i$, and let $C^i_d$ and $C'^i_d$ denote the clockwise and counter-clockwise paths of $d\in O(\log n)$ opinions adjacent to $i$. If more than $2d+1$ agents ever gather at $i$, one agent must have started outside these paths and traversed one of them. Every internal opinion on its route must then have been productive at least twice.

Now consider a fixed independent set of at least $d/2$ opinions on one of these paths. Each opinion initially has one agent, so it becomes extinct with probability $1/2$ in its first productive interaction. For an opinion to be productive twice, this first interaction must instead increase its support. Since no interaction involves two opinions in the independent set, we can assign independent fair coins to these first interactions. Thus, the probability that an agent traverses the path is at most $2^{-d/2}$. A union bound over both directions and all opinions gives the desired bound throughout the process. We give the full analysis in \cref{sec:a ring to rule them all}.

\end{proof}

\paragraph{Maximal Support on Graphs of Maximum Degree $\Delta$}
For the generalization to graphs of maximum degree $\Delta$ and initial support at most $\hat{s}$, we show that decision events of opinions at pairwise distance at least four satisfy essentially the same joint tail bound as independent events. This allows us to prove that most opinions are decided before they can attract many agents. However, each opinion may still have up to $O(\Delta^3)$ nearby opinions whose decision events are strongly correlated with its own. This leads to the relatively large polynomial dependence on $\Delta$ in our bound (which we do not believe to be optimal). More precisely, we obtain the following result.

\begin{restatable}[simplified]{paperthm}{shatteringdos}
\label{prop:shatteringdos}
Let $G=(O,E)$ be an arbitrary graph of maximum degree at most $\Delta$, where $\Delta\geq 1$, and suppose that every opinion has initial support at most $\hat{s}$. For the \EdgeVOG and the productive \EdgeVOG processes, with high probability, every opinion has support at most $\tilde{O}(\min\{n,\Delta^6\hat{s}\})$ simultaneously for all steps $t\geq 0$.
\end{restatable}

\begin{proof}[Proof sketch]
Fix an opinion $i\in O$. We want to show that it is unlikely that $i$ belongs to a connected component in which many agents can accumulate. On the cycle, it was enough to find two opinions within distance $O(\log n)$ of $i$ that become decided before being productive twice. Extending this argument to arbitrary graphs and initial configurations introduces two main obstacles. First, if an opinion initially has more than one agent, it no longer becomes extinct with probability $1/2$ after its first productive step.
Second, when $\Delta\geq 3$, the number of connected subsets containing $i$ grows exponentially with their size and not linearly.
Thus, even if any fixed connected subset is unlikely to contain many agents, we cannot directly take a union bound over all such subsets as there are exponentially many.

We address the first obstacle by showing that, for any fixed opinion $i$, the probability that $i$ remains undecided after $\tau$ productive steps is at most $\Delta\hat{s}^2/\tau$. This follows by adapting the potential argument from \Cref{lemma:new_potential}.
In particular, we show that the expected number of productive steps (that involve $i$) before $i$ becomes decided is at most $\Psi_i(\Vec{S}(0))$.
As $i$ has at most $\Delta$ incident edges and the initial potential of each such edge is at most $\hat{s}^2$, it holds $\Psi_i(\Vec{S}(0)) \leq\Delta\hat{s}^2$.
The claim then follows from Markov's inequality.
Next, we consider the dependencies between these opinions:
If $I\subseteq O$ consists of opinions at pairwise distance at least four, then the probability that every opinion in $I$ remains undecided after $\tau$ productive steps is at most $(\Delta\hat{s}^2/\tau)^{|I|}$.
Thus, for the purpose of tail bounds, sufficiently distant opinions behave as if their decision events were independent. The precise statement and its proof are given in \cref{lemma:pkbound}.

To address the second obstacle, we introduce so-called \emph{witness sets}: Fix an opinion $i$ and parameters $\tau,t\geq 0$. Informally, the witness set $W_i^\tau(t)$ is the largest connected set containing $i$ in which every opinion has been productive at least $\tau$ times before step $t$; see \cref{def:witness} for the precise definition. Figure~\ref{fig:witness} shows three examples for different values of $\tau$.

Witness sets allow us to control how many agents can accumulate in a connected part of the active opinion graph. Any agent that holds an opinion in $W_i^\tau(t)$ at time $t$ but did not initially hold an opinion in this set must have crossed its outer boundary. This boundary contains at most $\Delta|W_i^\tau(t)|$ opinions, and if $i$ has been productive at least $\tau$ times, every opinion on the outer boundary has been productive fewer than $\tau$ times; otherwise, it would belong to the witness set as we defined them to be maximal. If $i$ has been productive fewer than $\tau$ times, the witness is $\{i\}$ and its total number of productive boundary interactions is less than $\tau$.

A crude upper bound on the number of agents that can enter the witness set is therefore $\tau\Delta|W_i^\tau(t)|$. Consequently, if both $\tau$ and $|W_i^\tau(t)|$ are small, then the total support inside the witness set is also small.
However, this deterministic bound is too weak for our purposes. We therefore strengthen it slightly: Whenever a productive interaction crosses the boundary, an agent enters the witness set with probability $1/2$ and leaves it with the same probability. Thus, the net influx behaves like a simple random walk. A martingale concentration bound then shows that, with high probability, the net number of agents entering the witness set is at most $\tilde{O}(\sqrt{\tau\Delta}\cdot|W_i^\tau(t)|)$; see \cref{lemma:upperboundS}.

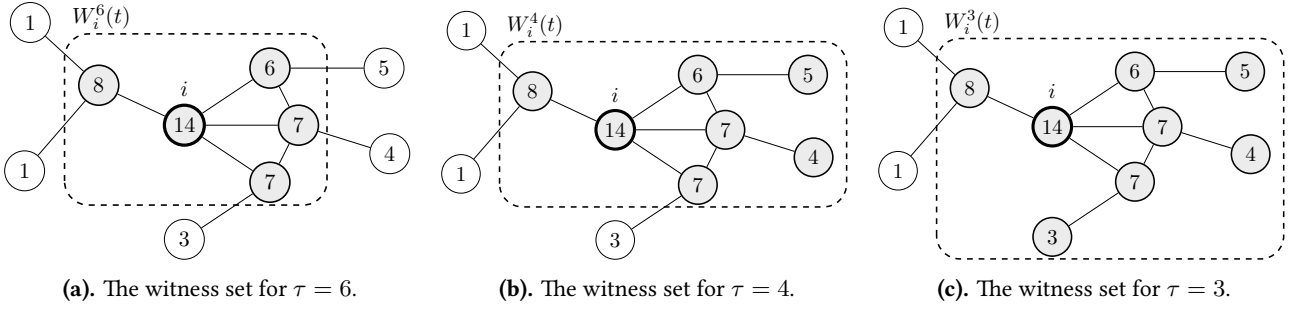
\begin{figure}[t]
\centering

\begin{subfigure}{0.32\textwidth}
    \centering
\resizebox{\linewidth}{!}{\begin{tikzpicture}[
    node/.style={circle, draw, minimum size=7mm, inner sep=0pt},
    witness/.style={circle, draw, thick, fill=gray!15, minimum size=7mm, inner sep=0pt},
    every edge/.style={draw},
    scale=1, transform shape
]

\node[witness, ultra thick, label=above:$i$] (i) at (0,0) {$14$};
\node[witness] (a) at (1.5,1) {$6$};
\node[witness] (b) at (2,0) {$7$};
\node[witness] (c) at (1.5,-1) {$7$};
\node[witness] (d) at (-1.5,0.7) {$8$};

\node[node] (x1) at (3.5,1) {$5$};
\node[node] (x2) at (3.6,-0.5) {$4$};
\node[node] (x3) at (-2.7,1.8) {$1$};
\node[node] (x4) at (-2.8,-0.8) {$1$};
\node[node] (x5) at (0,-2.0) {$3$};

\draw (i) -- (a);
\draw (i) -- (b);
\draw (i) -- (c);
\draw (i) -- (d);
\draw (a) -- (b);
\draw (b) -- (c);

\draw (a) -- (x1);
\draw (b) -- (x2);
\draw (d) -- (x3);
\draw (c) -- (x5);
\draw (x4) -- (d);

\begin{scope}[on background layer]
    \draw[rounded corners=12pt, thick, dashed]
        (-2.1,-1.4) rectangle (2.5,1.6);
\end{scope}

\node[anchor=south west] at (-2.1,1.55) {$W_i^{6}(t)$};

\end{tikzpicture}
}

    \caption{The witness set for $\tau=6$.}
\end{subfigure}
\hfill
\begin{subfigure}{0.32\textwidth}
    \centering
\resizebox{\linewidth}{!}{\begin{tikzpicture}[
    node/.style={circle, draw, minimum size=7mm, inner sep=0pt},
    witness/.style={circle, draw, thick, fill=gray!15, minimum size=7mm, inner sep=0pt},
    every edge/.style={draw},
    scale=1, transform shape
]

\node[witness, ultra thick, label=above:$i$] (i) at (0,0) {$14$};
\node[witness] (a) at (1.5,1) {$6$};
\node[witness] (b) at (2,0) {$7$};
\node[witness] (c) at (1.5,-1) {$7$};
\node[witness] (d) at (-1.5,0.7) {$8$};
\node[witness] (x1) at (3.5,1) {$5$};
\node[witness] (x2) at (3.6,-0.5) {$4$};

\node[node] (x3) at (-2.7,1.8) {$1$};
\node[node] (x4) at (-2.8,-0.8) {$1$};
\node[node] (x5) at (0,-2.0) {$3$};

\draw (i) -- (a);
\draw (i) -- (b);
\draw (i) -- (c);
\draw (i) -- (d);
\draw (a) -- (b);
\draw (b) -- (c);

\draw (a) -- (x1);
\draw (b) -- (x2);
\draw (d) -- (x3);
\draw (c) -- (x5);
\draw (x4) -- (d);

\begin{scope}[on background layer]
    \draw[rounded corners=12pt, thick, dashed]
        (-2.1,-1.4) rectangle (4.2,1.6);
\end{scope}

\node[anchor=south west] at (-2.1,1.55) {$W_i^{4}(t)$};

\end{tikzpicture}
}

    \caption{The witness set for $\tau=4$.}
\end{subfigure}
\hfill
\begin{subfigure}{0.32\textwidth}
    \centering
\resizebox{\linewidth}{!}{\begin{tikzpicture}[
    node/.style={circle, draw, minimum size=7mm, inner sep=0pt},
    witness/.style={circle, draw, thick, fill=gray!15, minimum size=7mm, inner sep=0pt},
    every edge/.style={draw},
    scale=1, transform shape
]

\node[witness, ultra thick, label=above:$i$] (i) at (0,0) {$14$};
\node[witness] (a) at (1.5,1) {$6$};
\node[witness] (b) at (2,0) {$7$};
\node[witness] (c) at (1.5,-1) {$7$};
\node[witness] (d) at (-1.5,0.7) {$8$};
\node[witness] (x1) at (3.5,1) {$5$};
\node[witness] (x2) at (3.6,-0.5) {$4$};
\node[witness] (x5) at (0,-2.0) {$3$};

\node[node] (x3) at (-2.7,1.8) {$1$};
\node[node] (x4) at (-2.8,-0.8) {$1$};

\draw (i) -- (a);
\draw (i) -- (b);
\draw (i) -- (c);
\draw (i) -- (d);
\draw (a) -- (b);
\draw (b) -- (c);

\draw (a) -- (x1);
\draw (b) -- (x2);
\draw (d) -- (x3);
\draw (c) -- (x5);
\draw (x4) -- (d);

\begin{scope}[on background layer]
    \draw[rounded corners=12pt, thick, dashed]
        (-2.1,-2.4) rectangle (4.2,1.6);
\end{scope}

\node[anchor=south west] at (-2.1,1.55) {$W_i^{3}(t)$};

\end{tikzpicture}
}

    \caption{The witness set for $\tau=3$.}
\end{subfigure}

\caption{Witness sets of a fixed opinion $i$ for three values of $\tau$. The number in each node denotes how often the opinion has been productive. As $\tau$ decreases, the witness set can only grow.}
\label{fig:witness}

\end{figure}

Given these observations, it remains to find a sufficiently small value $\tau^\star$ for which every witness set $W_i^{\tau^\star}(t)$ is small, simultaneously for all opinions $i$ and all times $t$. In \cref{lemma:Ctau}, we show that one can choose $\tau^\star=O(\Delta^5\hat{s}^2)$ such that every witness set contains at most $O(\Delta^3\log n)$ opinions with high probability.
The main ingredient of the proof is the joint tail bound above. Let $I$ be a set of opinions at pairwise distance at least four. The probability that every opinion in $I$ remains undecided after $\tau$ productive steps is at most $(\Delta\hat{s}^2/\tau)^{|I|}$. A textbook counting argument shows that, for each fixed opinion $i$, there are at most $O(\Delta^{4c\log n})$ relevant candidate witness sets containing a $4$-separated set $I$ of size $c\log n$. Choosing $\tau=c'\Delta^5\hat{s}^2$ for a sufficiently large constant $c'$, the probability that all opinions in such a set remain undecided is at most $(1/(c'\Delta^4))^{|I|}$. A union bound over all candidate sets therefore shows that, with high probability, no witness set contains a $4$-separated set of size at least $c\log n$.
Now let $I$ be a maximal $4$-separated subset of a witness set. By maximality, every opinion in the witness set is within distance three of some opinion in $I$. Since every radius-three neighborhood contains $O(\Delta^3)$ opinions, we obtain $|W_i^{\tau^\star}(t)|=O(\Delta^3|I|)=O(\Delta^3\log n)$.

Finally, the witness set initially contains at most $\hat{s}|W_i^{\tau^\star}(t)|$ agents, and its net influx is at most $\tilde{O}(\sqrt{\tau^\star\Delta}\cdot|W_i^{\tau^\star}(t)|)$. Substituting $\tau^\star=O(\Delta^5\hat{s}^2)$ and $|W_i^{\tau^\star}(t)|=O(\Delta^3\log n)$ gives a total support of $\tilde{O}(\Delta^6\hat{s})$. Since the support is always at most $n$, this proves the claimed bound. The full proof, including the details omitted here, is given in \Cref{sec:support}.
\end{proof}

\subsection{Conclusion and Future Work}
\label{sec:conclusion}

We studied the voter model with opinion graph. Our goal was to understand how the structure of the opinion graph affects the stabilization time, the number of surviving opinions, and the largest support attained by an opinion. For arbitrary opinion graphs and initial configurations, both the ordinary and the productive EdgeVOG processes stabilize within $\widetilde O(n^2)$ steps in expectation and with high probability. For the productive process, we further showed that the expected stabilization time is exactly the initial number $|P(0)|$ of productive pairs. We also proved that an opinion $i$ survives with probability at least $s_i/(s_i+\sum_{j\in N_i}s_j)$ and established a lower-tail bound for the total number of surviving opinions. In particular, on a $\Delta$-regular opinion graph with equally supported opinions, $\Omega(q/\Delta)$ opinions survive with probability $1-\exp(-\Omega(q/\Delta^3))$. Finally, if every opinion initially has support at most $\widehat s$, then the support attained by any opinion is at most $\widetilde O(\min{n,\Delta^6\widehat s})$ with high probability. Together with our parameterized stabilization bound, this shows that sparse opinion graphs with well-spread initial configurations can stabilize much faster than the worst-case bound while preserving many opinions and preventing any single opinion from becoming dominant.

The most immediate open problem is to improve the dependence on $\Delta$ and $\widehat s$, as we do not believe that our bounds are optimal.
Another direction would be to replace the maximum degree by finer structural parameters of the opinion graph, similar to our results on the survival probability.
It would also be interesting to determine which of our quantitative results extend to the arbitrary interaction graphs considered by Lanchier and Scarlatos. Further directions include dynamic or asymmetric opinion graphs and a more direct comparison with continuous bounded-confidence processes such as the Deffuant--Weisbuch model.

\newpage

\section{Proofs}
\label{sec:proofs}

\subsection{Proof of \cref{lemma:new_potential} (Properties of our Potential Function)}
\label{sec:potential}

In this section, we define the potential function $\Psi(\Vec{S}(t))$ which forms the basis of our analysis and prove some helpful properties.
 For each edge $\{i,j\} \in E$ of the opinion graph, let $\Psi_{ij}(\Vec{S}(t)) = S_i(t)\cdot S_j(t)$ be the product of the endpoints' supports in step $t$. We define

\begin{equation}
\label{eq:basic_potential}
    \Psi(\Vec{S}(t)) = \sum_{\{i,j\} \in E} \Psi_{ij}(\Vec{S}(t)).
\end{equation}

Recall that an opinion is decided if (and only if) $\Psi_i(\Vec{S}(t))=0$. Thus, the convergence time can be restated in terms of this potential as the random variable
$T \coloneqq \min \{ t \geq 0 \mid \Psi(\Vec{S}(t)) = 0\}.$
Note that $\Psi(\Vec{S}(t)) \le n(n-1)/2 \in O(n^2)$, the number of all unordered pairs of agents.
For the proofs, we bound the expected change from $\Psi(\Vec{S}(t))$ to $\Psi(\Vec{S}(t+1))$ for the \EdgeVOG process and then apply appropriate drift theorems to derive results on the expected convergence time. The following lemma bounds the expected difference in the potential from step $t$ to step $t+1$. As we use it for several theorems, we state in a slightly technical way.
\newpotential*
\begin{proof}
\newcommand{\Eij}[1]{\underset{\scriptscriptstyle e(t)= \{i,j\}}{\mathbf{E}}\left[#1\right]}
Throughout the proof, we omit the conditioning on $\{\Vec{S}(t) = \Vec{s}\}$ and let $\Vec{S}' = \Vec{S}(t+1)$ for readability. Further, if the process is converged, then trivially, $\E{\Psi(\Vec{S}')} = \Psi(\Vec{s}) = 0$.
Thus, assume the process is not converged, and in step $t$, a productive pair $\{v,w\} \in A \times A$ is chosen.
Finally, assume that $o_v = i$ and $o_w = j$.
Since the behavior of the process is the same no matter which particular agents with opinions $i$ and $j$ we choose, we simply define the event $e(t) = \{i,j\}$, which indicates that agents with opinion $i$ and $j$ are chosen.
First, we prove the following helpful claim:
\begin{claim}
\label{claim:casedist}
    $\E{\Psi_{kl}(\Vec{S}') \growingmid e(t) = \{i,j\}} = \begin{cases}
        \Psi_{kl}(\vec{s})-1 & \textit{if } \{k,l\} = \{i,j\}\\
        \Psi_{kl}(\vec{s}) & \textit{else.}
    \end{cases}$
\end{claim}
\begin{proof}

In the following, assume that the \ddw process chooses two agents on opinions $i$ and $j$ connected by an edge $\{i,j\}$ in step $t$.
It is easy to verify that if $e(t) = \{i,j\}$, with probability $\tfrac{1}{2}$ an agent from $i$ moves to $j$ or vice versa.
We will now examine the expected change of every summand of $\Psi(\Vec{S}') = \sum_{\{k,l\} \in E} \Psi_{kl}(\Vec{S}')$.
First, we consider the summand $\Psi_{ij}(\Vec{S}')$ for the edge $\{i,j\} \in E$ with $e(t) = \{i,j\}$.
\begin{align*}
    \Eij{\Psi_{ij}(\Vec{S}')} &=\frac{1}{2}(s_i-1)(s_j+1)+\frac{1}{2}(s_i+1)(s_j-1)=s_i s_j-1.
\end{align*}

Second, we consider a summand $\Psi_{il}(t)$ for edge $\{i,l\} \in E$ with $l \neq j$.
Here, only opinion $i$ is part of the active edge. Thus, the support of $i$ increases or decreases by $1$, whereas the support of $l$ remains unchanged.
Therefore,
\begin{align*}
    \Eij{\Psi_{il}(\Vec{S}')} &= \frac{1}{2}\bigl((s_i-1)s_l+(s_i+1)s_l\bigr)=s_i s_l.
\end{align*}
For $\Psi_{kj}(t)$ with $k\neq i$, the calculation is equivalent.
Finally, we consider $\Psi_{kl}(t)$ for some edge $\{ k,l \}$ with $\{ k,l \} \cap \{ i,j \} = \emptyset$. This potential is trivially unchanged as no endpoint changes its support.
\end{proof}
\noindent Now, we use the law of total expectation and linearity of expectation to prove the lemma. We define $\rho(t) \coloneqq \pr{\mathcal{P}(t) = 1} = \sum_{\{i,j\} \in E}\pr{e(t) = \{i,j\}}$ as the probability of an productive step.
\begin{align}
\label{eq:psihelper}
     \E{\Psi(\Vec{S}')} &= (1-\rho(t)) \cdot \Psi(\Vec{s}) + \rho(t)\E{\Psi(\Vec{S}') \growingmid \mathcal{P}(t) = 1}.
\end{align}
\newcommand{\Et}[1]{\underset{\scriptscriptstyle \mathcal{P}(t) = 1}{\mathbf{E}}\left[#1\right]}
\newcommand{\prt}[1]{\underset{\scriptscriptstyle \mathcal{P}(t) = 1}{\mathbf{Pr}}\left[#1\right]}
Now use the linearity of expectation to get
\begin{align*}
     \Et{\Psi(\Vec{S}')}
     &= \Et{\sum_{\{k,l\} \in E}\Psi_{kl}(\Vec{S}') }
     =\sum_{\{i,j\} \in E} \prt{ e(t) = \{i,j\} } \left(\sum_{\{k,l\} \in E } \Eij{\Psi_{kl}(\Vec{S}')}\right).
\end{align*}
Next, we apply \Cref{claim:casedist} to bound each $\E{\Psi_{kl}(\Vec{S}') \growingmid e(t) = \{i,j\}}$ for all $\{k,l\}, \{i,j\} \in E$:
\begin{align*}
        \Et{\Psi(\Vec{S}')}&=\sum_{\{i,j\} \in E} \prt{e(t) = \{i,j\}} \left( {\Psi_{ij}(t)} - 1 + \sum_{\{k,l\} \in E \setminus \{\{i,j\}\}} \Psi_{kl}(\Vec{s})\right)\\
     &= (\Psi(\Vec{s})-1)\sum_{\{i,j\} \in E} \prt{e(t) = \{i,j\}} = (\Psi(\Vec{s})-1) .
\end{align*}
Substituting this in \Cref{eq:psihelper} allows to conclude
\begin{align*}
      \E{\Psi(\Vec{S}')} &= (1-\rho(t))\Psi(\Vec{s}) + (\Psi(\Vec{s})-1) \rho(t)  = \Psi(\Vec{s})-\rho(t) =  \Psi(\Vec{s})-\pr{\mathcal{P}(t) = 1}. \qedhere
\end{align*}

\end{proof}

\subsection{Proof of \cref{thm:newPotential} (Stabilization Time of the Ordinary Process)}
\label{sec:general graph bound}
In this section, we prove \Cref{thm:newPotential}.
For the first statement, in addition to \Cref{lemma:new_potential}, we use the following multiplicative drift theorem~\cite{DoerrJW12, Johannsen2010,DBLP:conf/gecco/RoweS12} in the version by Lengler~\cite{Lengler2020}.
\begin{theorem}[{Theorem 2.3.11 in \cite{Lengler2020} (Based on \cite{DoerrJW12,Johannsen2010,DBLP:conf/gecco/RoweS12}})]
\label{driftthm:multiplicative}
Let $(X_t)_{t\geq 0}$ be a sequence of non-negative random variables with a finite state space $\mathcal{S} \subseteq \mathbb{R}_0^+$ such that $0 \in \mathcal S$. Let $s_{\text{min}} := \min(\mathcal{S} \setminus \{0\})$, let $T := \inf\{t \geq 0 \mid X_t =0\}$, and for $t \geq 0$ and $s \in \mathcal S$ let $\partial(s) := E[X_{t} - X_{t+1} \mid X_t = s]$. Suppose there exists $\delta > 0$ such that for all $s \in \mathcal{S} \setminus \{0\}$ and all $t \ge 0$ the drift is $\partial_t(s) \ge \delta s$. Then, $\E{T} \leq \frac{1 + \E{\ln(X_0/s_{\text{min}})]}}{\delta}$.
\end{theorem}

The expected convergence time for the ordinary \EdgeVOG process follows from using \Cref{driftthm:multiplicative} to bound the time when $\Psi(\Vec{S}(t))$ becomes $0$.
Note that $X_t$ here becomes the value of our potential $\Psi(\Vec{S}(t))$. This implies $X_0 \le {n \choose 2}$ and $s_{\min} = 1$. The probability to make a productive step in the global process is $\Psi(\Vec{S}(t))/{n \choose 2} \ge \Psi(\Vec{S}(t))/n^2$.
Note that, by the law of total expectation, \Cref{lemma:new_potential} implies that for every value $\psi \in [0,n(n-1)]$ we have
\begin{align*}
    \E{\Psi(\Vec{S}(t+1)) \growingmid \Psi(\Vec{S}(t)) = \psi } = \psi - \E{\mathcal{P}(t) \growingmid \Psi(\Vec{S}(t)) = \psi}.
\end{align*}
Therefore, for a fixed value $\Psi(\Vec{S}(t)) = \psi$, \Cref{lemma:new_potential} implies:
\begin{align}
    \label{eq:globalMulti}
    \E{\Psi(\Vec{S}(t+1)) \growingmid \Psi(\Vec{S}(t)) = \psi} \leq \psi - \frac{\psi}{n^2} \leq \psi\left(1-\frac{1}{n^2}\right).
\end{align}
Thus, we can use \cref{driftthm:multiplicative} with $\delta = \frac{1}{n^2}$ and obtain an expected convergence time of $O(n^2 \log n)$.

For the high probability bound, we can use (arguably) even simpler observations.
Write $N=\binom n2$. Since the probability of a productive step is exactly $\Psi(\Vec S(t))/N$, the law of total expectation gives
\begin{align*}
    \E{\Psi(\Vec S(t+1))}
    &=\E{\Psi(\Vec S(t))-\frac{\Psi(\Vec S(t))}{N}} =\left(1-\frac1N\right)\E{\Psi(\Vec S(t))}.
\end{align*}
An elementary induction and $1-1/N\leq e^{-1/N}$ therefore yield
\[
    \E{\Psi(\Vec S(t))}\leq\Psi(\vec s_0)e^{-t/N}\leq n^2e^{-t/N}.
\]
For any fixed $c>1$, choosing $t=\lceil(c+2)N\log n\rceil$ makes this expectation at most $n^{-c}$.
As the potential is a nonnegative integer and is positive exactly before stabilization, Markov's inequality gives
\[
    \pr{T>t}=\pr{\Psi(\Vec S(t))\geq1}\leq n^{-c}.
\]
This proves the high-probability bound in \Cref{thm:newPotential}.

\newpage
\subsection{Proof of \Cref{thm:ProdWorstCase} (Worst-Case Stabilization Time of the Productive Process)}
\label{sec:wcproductive}
First, we note that the ordinary \EdgeVOG process always requires as least as many steps as the productive \EdgeVOG process (as the latter simply ignores the non-productive steps). Hence, with high probability, it also requires $\tO{n^2}$ steps.
Therefore, we will only prove the improved bound on the expected convergence time of the productive \EdgeVOG process.
For the proof, we use the following drift theorem.
\begin{theorem}[Theorem 2.3.1 in \cite{Lengler2020} (Based on \cite{HeYao:04:drift})]\label{driftthm:additive}
Let $(X_t)_{t\geq 0}$ be a sequence of non-negative random variables with a finite state space $\mathcal{S} \subseteq \mathbb{R}_0^+$ such that $0 \in \mathcal S$. Let $T := \inf\{t \geq 0 \mid X_t =0\}$.
 Suppose that there exists $\delta > 0$ such that for all $s\in \mathcal S \setminus\{0\}$ and for all $t\geq 0$, we have $\partial_t(s):=\E{X_t-X_{t+1}\mid X_t=s} = \delta$.
Then,
$\E{T} = \frac{\E{X_0}}{\delta}.$
\end{theorem}
The original formulation in \cite{Lengler2020} is slightly more general. We simplified it for our purposes\footnote{Lengler also gives specific bounds when only an upper or lower bound of $\delta$ is known. However, as we can compute the exact value, we omitted these cases here.}.
That being said, consider the potential function $\Psi(\Vec{S}(t))$ and the convergence time $T \coloneqq \min \{t \geq 0: \Psi(\Vec{S}(t)) = 0 \}$. Our goal is to bound $\E{T}$ using \Cref{driftthm:additive}. Towards this end, we bound the term
\begin{align*}
    \partial_t(s):=\E{\Psi(\Vec{S}(t))-\Psi(\Vec{S}(t+1))\growingmid \Psi(\Vec{S}(t)) = s}
\end{align*}
for all possible values $s > 0$.
Note that the productive \EdgeVOG process picks exactly one pair of agents per step as long as there are productive pairs and none otherwise.
This means that we can apply \Cref{lemma:new_potential} to bound the change in potential. The  lemma states that the potential decreases by $\E{\mathcal{P}(t) \growingmid \Vec{S}(t) = \Vec{s}}$ in expectation where $\mathcal{P}(t)$ is a random variable which equals 1 if step $t$  is productive.
Clearly, as long as there are productive pairs, we have  $\Psi(\Vec{S}(t)) > 0$.
Thus, the probability of a productive step is $1$, which means $\E{\mathcal{P}(t) \mid \Psi(\Vec{S}(t)) > 0} = 1$.
Therefore, by \Cref{lemma:new_potential} it holds for all $s > 0$ that
\begin{align}
    \partial_t(s) & = \E{\Psi(\Vec{S}(t))-\Psi(\Vec{S}(t+1))\mid \Psi(\Vec{S}(t)) = s}
    = s - \E{\Psi(\Vec{S}(t+1))\mid \Psi(\Vec{S}(t)) = s} \nonumber \\
    &= s -  \left(s - \E{\mathcal{P}(t) \growingmid \Psi(\Vec{S}(t)) = s}\right) = s-(s-1) = 1.\label{eqn:proddiff}
\end{align}
Thus, we have $\delta = 1$ for all $s > 0$, so \Cref{driftthm:additive} yields an expected convergence time of $\E{T} = \Psi(\Vec{s}_0) = |P(0)|$, and the proposition follows.

    Consider a star graph of $n$ opinions/nodes, with one agent holding each opinion.
This configuration has $n-1$ productive pairs.
The process stabilizes precisely when the support of the center opinion reaches $0$ or $n$.
In each step, exactly one agent either adopts or abandons the center opinion, with both events occurring with equal probability.
Consequently, the support of the center evolves as a one-dimensional random walk with absorbing barriers at $0$ and $n$.
By \Cref{lem:gr_tail_lower}, the process requires $\Omega(n^2)$ steps with probability $\Omega(1/n)$.

\subsection{Proof of \Cref{thm:newPotentialProd} (Parameterized Stabilization Time of the Productive Process)}
\label{sec:whpproductive}

\newcommand{\tPsi}{\ensuremath{\tilde{\Psi}}}

Before we start with the proof, we begin with some helpful definitions and observations.

For the proof, write $\hat{\mathcal S}=\hat{\mathcal S}_a(G,\vec s_0)$ for a sufficiently large fixed failure exponent $a$, and let $b\coloneqq\Delta\hat{\mathcal S}$.
If the initial configuration is already stable, there is nothing to prove. Otherwise, $\Delta\geq1$ and $b\geq1$.
Let
\[
    \tau_{\mathrm{supp}}\coloneqq\inf\{t\geq0:\max_{i\in O}S_i(t)>\hat{\mathcal S}\},
    \qquad
    \tilde{T}\coloneqq T\wedge\tau_{\mathrm{supp}}.
\]
By definition of $\hat{\mathcal S}$, $\pr{\tau_{\mathrm{supp}}<\infty}\leq n^{-a}$.
Rather than conditioning on the event that the support cap is never exceeded, we use the stopped potential
\[
    \tilde\Psi(\Vec S(t))\coloneqq\Psi(\Vec S(t\wedge\tau_{\mathrm{supp}})).
\]
Thus, $\tilde\Psi$ freezes at the first support violation, while $\tilde{T}$ is the first time when the process either stabilizes or violates the cap.
The original process and the stopped process agree until $\tilde{T}$. Hence, a high-probability bound for $\tilde{T}$ gives the same bound for $T$ after adding at most $n^{-a}$ to the failure probability.

\medskip

The remaining section is structured as follows.
On a high level, the proof proceeds in two phases, depending on whether the potential is large or small.
In Lemma~\ref{lemma:azumapotential} we first establish an auxiliary result that is used in both phases.
It shows that, for short time horizons,
the potential grows at most by a small constant factor.
Using this tool, we distinguish two phases.
{Let $c \geq 1$ be a constant that will be fixed in the analysis.}
In the first phase where $\tPsi(\Vec{S}(t)) \ge 32c\cdot\Delta^4\cdot b^2\log^2 n$,
we show that the potential decreases rapidly.
Proposition~\ref{prop:largepotential} proves that within $O(\tPsi(\Vec{s}_0) + b^2 \log^3 n)$ steps
the potential drops to $32c\cdot\Delta^4\cdot b^2\log^2 n$ with high probability.
The proof is based on a layered induction argument that partitions time into phases.
In each phase, a constant fraction of the remaining productive edges is eliminated.
In the second case, once $\tPsi(\Vec{S}(t)) \le 32c\cdot\Delta^4\cdot b^2\log^2 n$,
the argument no longer applies as there too few opinions left.
Instead, Proposition~\ref{prop:smallpotential} shows that in this regime
the process converges within $O(\Delta^4 \cdot b^2 \log^3 n)$ additional steps with high probability.
Combining both cases yields the desired convergence bound for the productive
\EdgeVOG process.

That being said, we begin with a technical result that we need for either case:
We show that the productive edges are highly concentrated if we do not observe the process for too long.
\begin{lemma}
\label{lemma:azumapotential}

Let $K=\Psi(\vec s_0)+b^2\log^2n$. Then, for every $r\geq1$ and integer $\tau\leq r^2K\log n$,
\[
    \pr{\max_{0\leq u\leq\tau}\tilde\Psi(\Vec S(u))
       >\Psi(\vec s_0)+8r^{3/2}K}
    \leq n^{-r}.
\]

\end{lemma}

The lemma follows from the maximal upper-tail version of Azuma's inequality (see \Cref{thm:azuma}) and the fact that the stepwise change of the stopped potential is bounded by $2b$.

\begin{proof}

Note that $(\tilde\Psi(\Vec S(t)))_{t\geq0}$ is a supermartingale, since it is obtained by stopping the original potential at $\tau_{\mathrm{supp}}$.
Next, consider its stepwise difference before $\tau_{\mathrm{supp}}$.
W.l.o.g. let $i$ receive an agent and let $j$ lose an agent in step $t$.
A direct calculation gives
\[
    \Psi(\Vec S(t+1))-\Psi(\Vec S(t))
    =\sum_{h\in N_i}S_h(t)-\sum_{h\in N_j}S_h(t)-1.
\]
Since every support before this interaction is at most $\hat{\mathcal S}$,
\[
    |\tilde\Psi(\Vec S(t+1))-\tilde\Psi(\Vec S(t))|
    \leq\Delta\hat{\mathcal S}+1\leq2b.
\]
This also holds for the interaction that first violates the support cap, since the bound uses the configuration before that interaction. After $\tau_{\mathrm{supp}}$, the difference is zero.
Thus, for $\tau\geq1$, maximal Azuma gives
\begin{align*}
    \pr{\max_{0\leq u\leq\tau}\tilde\Psi(\Vec S(u))
       >\Psi(\vec s_0)+8r^{3/2}K}
    &\leq\exp\left(-\frac{(8r^{3/2}K)^2}{2\tau(2b)^2}\right) \leq\exp\left(-\frac{8rK}{b^2\log n}\right)
    \leq n^{-8r}\leq n^{-r}.
\end{align*}
For $\tau=0$, the assertion is immediate. This proves the lemma.

\end{proof}

Equipped with this lemma, we now distinguish between the two cases mentioned earlier.

\subsubsection*{(Case 1) ${\Psi(\Vec{S}(t)) \geq 32c\cdot\Delta^4\cdot b^2\log^2 n}$}
In this first part, we consider only configuration with sufficiently large potentials.
We show that the potential rapidly decreases from its initial value of $\Psi(\Vec{s}_0)$ to $32c\cdot\Delta^4 \cdot b^2 \cdot \log^2 n$ within $O(\Psi(\Vec{s}_0) + b^2\log^3 n)$ steps, w.h.p.
Formally, we show that the following holds:
\begin{restatable}{prop}{largepotential}
\label{prop:case1}
\label{prop:largepotential}

For any fixed $c>3$, the productive \EdgeVOG process reaches a configuration with potential at most $32c\Delta^4b^2\log^2n$ within $O_c(\Psi(\vec s_0)+b^2\log^3n)$ steps w.h.p.

\end{restatable}
The proof uses a \emph{layered induction}.
We partition the execution into $O(\log n)$ phases where phase $i$ is of length at most $O(\tfrac{\Psi(\vec{s}_0)}{2^i} + b^2\log^2 n)$ w.h.p.
In other words, w.h.p., the durations decrease geometrically (except for an additive term).
For each phase, we show that enough productive interactions occur to eliminate at least half of the remaining productive edges with high probability (as long as the potential is at least $32c\cdot\Delta^4\cdot b^2\log^2 n$).
Iterating this argument over $O(\log n)$ phases yields an exponential decay in the potential until it reaches $32c\cdot\Delta^4\cdot b^2\log^2 n$.
Since, w.h.p., the phase lengths form a geometric sequence that sums to $O(\Psi(\vec s_0)+b^2\log^3n)$, we conclude that the potential has fallen to $32c\cdot\Delta^4\cdot b^2\log^2 n$, as claimed.

To prove this statement, we require an auxiliary technical result. First, we show that the potentials of $4$-separated sets satisfy a suitable product bound.
For $B\geq1$, define
\[
    \tau_B\coloneqq\inf\{t\geq0:\Psi(\Vec S(t))>B\}
\]
and the auxiliary node potentials
\[
    \hat\Psi_i^{(B)}(t)=
    \begin{cases}
        \Psi_i(\Vec S(t)),&\text{if }t<\tau_{\mathrm{supp}}\wedge\tau_B,\\
        0,&\text{otherwise}.
    \end{cases}
\]
Let $\hat\Psi^{(B)}(t)=\tfrac12\sum_{i\in O}\hat\Psi_i^{(B)}(t)$.
As long as neither cap has been violated, $\hat\Psi^{(B)}(t)=\Psi(\Vec S(t))$.
After a violation, all auxiliary node potentials remain zero. Unlike the stopped potential $\tilde\Psi$, these auxiliary potentials are killed at a support violation, so a positive frozen value cannot prevent their decay.

Given this useful definition, we show
\begin{lemma}
\label{lemma:prodpotential}
Consider an opinion graph $G = (O,E)$ and threshold $B \geq 1$.
Further, let $\mathcal{I} \subset O$ be an \textbf{arbitrary} $4$-separated set in $G$.
For each $t \geq 1$ define
\[
    \Phi^{(B)}_{\mathcal{I}}(t) = \prod_{i \in \mathcal{I}} \hat{\Psi}^{(B)}_i(\Vec{S}(t)).
\]
Then, for all $t \geq 0$, it holds:
\begin{align*}
     \E{\Phi^{(B)}_\mathcal{I}(t)} \leq \prod_{i \in \mathcal{I}} \left(\Psi_i(\Vec{s}_0) \cdot (1-\tfrac{1}{B})^t \right).
\end{align*}
\end{lemma}

The proof is by induction over the time steps. Before a cap violation, at most one of the selected node potentials can change in an interaction. Killing all factors at a cap violation can only decrease their nonnegative product.

\begin{proof}

As $B$ and $\mathcal I$ are fixed, we omit them from the notation where possible.
The base case $t=0$ is immediate. If either cap has already been violated, all auxiliary potentials are zero and the upper bound holds trivially.
For the induction step, condition on the full history $\mathcal F(t)$ and suppose that the auxiliary product is positive. In particular, neither cap has been violated, and every opinion in $\mathcal I$ is undecided.
Write $\psi=\Psi(\Vec S(t))\leq B$.

First ignore the killing that may occur in step $t$.
Since $\mathcal I$ is $4$-separated, the selected edge can affect at most one node potential in the product.
By \Cref{claim:casedist}, $\Psi_i$ decreases by one in expectation conditional on an interaction involving $i$, and does not change in expectation for all other interactions.
The probability that $i$ is productive is $\Psi_i(\Vec S(t))/\psi$.
Therefore,
\begin{align}
    \E{\prod_{i\in\mathcal I}\Psi_i(\Vec S(t+1))\mid\mathcal F(t)}
    &=\left(1-\frac{|\mathcal I|}{\psi}\right)
       \prod_{i\in\mathcal I}\Psi_i(\Vec S(t)).
    \label{eq:important}
\end{align}
If a cap is violated in this interaction, the auxiliary product is set to zero, which only decreases it.
Moreover, when the product is positive, the incident edge sets of the opinions in $\mathcal I$ are disjoint and each has positive integer potential. Hence, $|\mathcal I|\leq\psi\leq B$.
Using Bernoulli's inequality, we obtain
\begin{align*}
    \E{\Phi_{\mathcal I}^{(B)}(t+1)\mid\mathcal F(t)}
    &\leq\left(1-\frac{|\mathcal I|}{B}\right)\Phi_{\mathcal I}^{(B)}(t) \leq\left(1-\frac1B\right)^{|\mathcal I|}\Phi_{\mathcal I}^{(B)}(t).
\end{align*}
If the current product is zero, it remains zero: either a cap was violated or one of its opinions is already decided. Thus, the same inequality holds in that case.
Taking expectations and iterating gives
\[
    \E{\Phi_{\mathcal I}^{(B)}(t)}
    \leq\left(1-\frac1B\right)^{t|\mathcal I|}
       \prod_{i\in\mathcal I}\Psi_i(\vec s_0),
\]
which proves the lemma. For the empty set, both products are one.

\end{proof}

The second auxiliary result lets bound the decrease in potential.
\begin{lemma}
\label{lemma:inductionstep}

Consider the productive \EdgeVOG process at a time $t<\tilde{T}$, and write $\psi=\Psi(\Vec S(t))$.
For every fixed $c\geq1$, start a phase at time $t$ and end it when the potential first falls to at most  $ \max\{\psi/2,32c\Delta^4b^2\log^2n\}$
(or when the support bound is violated).
Conditional on $\mathcal F(t)$, the phase lasts $O(c\cdot \psi+b^2\log^2n)$ steps with probability at least $1-n^{-c}$.

\end{lemma}
\begin{proof}

We begin by introducing some notation and simplifying conventions.
Condition on the history up to time $t$, and shift the time index so that this configuration is at time zero. Write $\vec s_0$ for this configuration and $\psi=\Psi(\vec s_0)$.
For the fixed constant $c\geq1$, let
\[
    K=\psi+b^2\log^2n,\qquad r=10^4(c+1)^2,
    \qquad B=(8r^{3/2}+1)K,
    \qquad\tau=\lceil B\log(32c)\rceil.
\]
These choices give $\tau\in O_c(K)$ and $\tau\leq r^2K\log n$ for $n\geq2$.
The caps and auxiliary potentials are restarted at the beginning of the phase, using the same support threshold $\hat{\mathcal S}$.

Partition $O$ into $4$-separated sets $I_1,\ldots,I_k$ by properly coloring the graph joining opinions at distance at most three.
The radius-three ball bound gives
\[
    k\leq\min\{q,2\Delta^3\}\leq\min\{n,2\Delta^4\}.
\]
For each such set $I$, write
$\hat\Psi_I^{(B)}(u)=\sum_{i\in I}\hat\Psi_i^{(B)}(u)$
and $\psi_I=\sum_{i\in I}\Psi_i(\vec s_0)$.
We first show that
\begin{align}
    \pr{\hat\Psi_I^{(B)}(\tau)\geq
          \psi_I/4+8cb^2\log n}\leq n^{-8c}.
    \label{eq:netbound}
\end{align}

By \Cref{lemma:prodpotential}, for every subset $I'\subseteq I$,
\begin{align*}
    \E{\prod_{i\in I'}\hat\Psi_i^{(B)}(\tau)}
    &\leq\left(1-\frac1B\right)^{\tau|I'|}
           \prod_{i\in I'}\Psi_i(\vec s_0) \leq(32c)^{-|I'|}\prod_{i\in I'}\Psi_i(\vec s_0).
\end{align*}
Every auxiliary node potential lies in $[0,b^2]$, because before killing its value is at most $\Delta\hat{\mathcal S}^2\leq b^2$.
Thus, \Cref{thm:scheidelbound} applies with upper bound $b^2$, $p=1/(32c)$, $x_i=\Psi_i(\vec s_0)$, and $\mu=\psi_I/(32c)$.
If $\mu=0$, all the relevant auxiliary potentials vanish almost surely.
Otherwise, set
\[
    a_I=\psi_I/4+8cb^2\log n,\qquad\delta=a_I/\mu.
\]
Then $\delta\geq8c\geq e^2$, and consequently
\[
    \pr{\hat\Psi_I^{(B)}(\tau)\geq a_I}
    \leq e^{-a_I/b^2}\leq n^{-8c}.
\]
This proves \eqref{eq:netbound} without any lower bound on $\psi_I$ or $|I|$.

Taking a union bound over the at most $n$ color classes, with failure probability at most $n^{1-8c}$, gives
\begin{align*}
    \hat\Psi^{(B)}(\tau)
    &=\frac12\sum_{j=1}^k\hat\Psi_{I_j}^{(B)}(\tau) \leq\frac14\Psi(\vec s_0)+4ckb^2\log n\\
    &\leq\frac\psi4+8c\Delta^4b^2\log n \leq\frac\psi4+16c\Delta^4b^2\log^2n
    \leq\max\{\psi/2,32c\Delta^4b^2\log^2n\}.
\end{align*}
Here, we used $\log n\geq\log2>1/2$.

By \Cref{lemma:azumapotential}, the stopped potential exceeds $B$ during the first $\tau$ steps with probability at most $n^{-r}$.
Suppose that the stopped potential has not crossed $B$ during these steps.
If the support bound has been violated by time $\tau$, the phase has already ended.
Otherwise, $\hat\Psi^{(B)}(\tau)=\Psi(\Vec S(\tau))$, so the preceding bound on the auxiliary potential implies that the phase has ended by time $\tau$ as well.
The total failure probability is therefore at most
$n^{1-8c}+n^{-r}\leq n^{-c}$.
This proves the induction step.

\end{proof}
Equipped with lemma, we can now proceed to the proof of \Cref{prop:case1}.
\begin{proof}[Proof of \Cref{prop:case1}]

We need a final set of definitions.
Fix $c>3$ and let $H=32c\Delta^4b^2\log^2n$.
If $\Psi(\vec s_0)\leq H$, the assertion is immediate.
Otherwise, let
\[
    L=\left\lceil\log_2\frac{\Psi(\vec s_0)}H\right\rceil,
    \qquad h_j=\max\{\Psi(\vec s_0)/2^j,H\},\quad 0\leq j\leq L.
\]
Set $\tau_0=0$. For $j\geq1$, let $\tau_j$ be the first time at or after $\tau_{j-1}$ at which either $\tilde{T}$ has been reached or the potential is at most $h_j$.
On a phase that starts before $\tilde{T}$, \Cref{lemma:inductionstep}, conditional on its starting history, implies
\[
    \tau_j-\tau_{j-1}\in O_c(h_{j-1}+b^2\log^2n)
\]
with probability at least $1-n^{-c}$.
If the process has already stopped, the phase has length zero.
Since $L\in O(\log n)$, the union bound shows that all phases satisfy these bounds w.h.p.
On this event,
\begin{align*}
    \tau_L
    &\leq\sum_{j=1}^L O_c\left(\frac{\Psi(\vec s_0)}{2^{j-1}}+b^2\log^2n\right) =O_c\bigl(\Psi(\vec s_0)+b^2\log^3n\bigr).
\end{align*}
At this time, either $\tilde{T}$ has been reached or the potential has fallen to $H$.
A support violation occurs with probability at most $n^{-a}$, so the same conclusion holds for the original process w.h.p.

\end{proof}

\subsubsection*{Case 2: $\Psi(\Vec{S}(t)) \leq 32c\cdot\Delta^4\cdot b^2\log^2 n$.}
For the last part of the analysis, we consider the process starting at configurations with a potential of at most $32c\Delta^4 b^2 \log^2 n$.
In this case, our layered induction argument falls apart as only few active opinions are left.
Instead, we argue that w.h.p. the process converges after $O( \Delta^4 b^2 \log^3 n)$ additional steps.
\begin{restatable}{prop}{smallpotential}
\label{prop:smallpotential}

Suppose $\Psi(\vec s_0)\leq32c\Delta^4b^2\log^2n$ for a fixed $c\geq1$.
Then, $\tilde{T}\in O_c(\Delta^4b^2\log^3n)$ w.h.p. Hence, the original productive process satisfies the same bound when the probability of a support violation is added to the failure probability.

\end{restatable}
\begin{proof}

We introduce some notation and simplifying conventions.
For the fixed $c\geq1$, let
\[
    K=\Psi(\vec s_0)+b^2\log^2n,\qquad r=10^4(c+1)^2,
    \qquad B=(8r^{3/2}+1)K,
    \qquad\tau=\lceil B(c+4)\log n\rceil.
\]
Then $\tau\in O_c(K\log n)$ and $\tau\leq r^2K\log n$.
Applying \Cref{lemma:prodpotential} to singleton sets and summing yields
\begin{align*}
    \E{\hat\Psi^{(B)}(\tau)}
    &\leq\Psi(\vec s_0)\left(1-\frac1B\right)^\tau
    \leq n^2e^{-\tau/B}
    \leq n^{-(c+2)}.
\end{align*}
The auxiliary global potential is either the integer-valued original potential or zero.
Thus, Markov's inequality gives
\[
    \pr{\hat\Psi^{(B)}(\tau)\geq1}\leq n^{-(c+2)}.
\]
By the maximal estimate in \Cref{lemma:azumapotential}, the stopped potential exceeds $B$ during this interval with probability at most $n^{-r}$.
On the event $\{\tilde{T}>\tau\}$, if no such crossing occurs, the auxiliary and original potentials agree and are positive at time $\tau$.
Consequently,
\[
    \pr{\tilde{T}>\tau}\leq n^{-(c+2)}+n^{-r}\leq n^{-c}.
\]
Since $K\in O_c(\Delta^4b^2\log^2n)$ under the assumption of the proposition, this proves the claim.

The same calculation, without the small-potential assumption, gives
$\tilde{T}\in O_c((\Psi(\vec s_0)+b^2\log^2n)\log n)$ w.h.p., which proves the first bound in \Cref{thm:newPotentialProd}.
For the second bound, apply \Cref{prop:largepotential} first and start the present argument at the first time the potential is at most $32c\Delta^4b^2\log^2n$.
All estimates hold conditional on that starting history before a support violation. Adding the original failure probability $n^{-a}$ once gives
\[
    T\in O_c\bigl(\Psi(\vec s_0)+\Delta^4b^2\log^3n\bigr)
\]
w.h.p. Substituting $b=\Delta\hat{\mathcal S}$ proves the second bound.

\end{proof}

\subsection{Proof of \Cref{thm:survival} (Survival of a Fixed Opinion)}
\label{sec:survival}
In this section, we show

\survival*

The proof of the theorem follows from the following two lemmas. The first
identifies the terminal value of $R_i(t)$, while the second establishes the
required submartingale property.

\begin{lemma}
\label{lem:terminal-value-R}
For every opinion $i$, we have $R_i(T)=X_i$ almost surely.
\end{lemma}

\begin{proof}
Recall that $\tau_i$ is the first time at which either $S_i(t)=0$ or
$M_i(t)=0$. If $S_i(\tau_i)=0$, then opinion $i$ is extinct and cannot
reappear, since an opinion can only be copied from an agent that currently
holds it. Hence, $X_i=0$.

Otherwise, $M_i(\tau_i)=0$ and $S_i(\tau_i)>0$. In this case, no opinion in
$N_i$ is present. Since these opinions are extinct and cannot reappear,
opinion $i$ can no longer participate in a productive interaction and
therefore survives. Hence, $X_i=1$. Consequently,
\begin{align*}
    R_i(\tau_i)
    &=
    \begin{cases}
        0, & \text{if } S_i(\tau_i)=0,\\
        1, & \text{if } M_i(\tau_i)=0 \text{ and } S_i(\tau_i)>0,
    \end{cases}
    =X_i.
\end{align*}
At the stabilization time $T$, every opinion is either extinct or has no
surviving neighboring opinion. Thus, $\tau_i\leq T$ almost surely. Since
$R_i(t)$ is stopped at $\tau_i$, it follows that
$R_i(T)=R_i(\tau_i)=X_i$.
\end{proof}

\begin{lemma}
\label{lem:R-submartingale}
For every opinion $i$, the process $(R_i(t))_{t\geq 0}$ is a bounded
submartingale with respect to the natural filtration
$(\mathcal F(t))_{t\geq 0}$.
\end{lemma}

\begin{proof}
By definition, $0\leq R_i(t)\leq 1$ for every $t\geq 0$, so the process is
bounded. Moreover, on the event $\{\tau_i\leq t\}$, the process has already
been stopped, and therefore $R_i(t+1)=R_i(t)$. We may thus fix a time $t$ and
work on the event $\{t<\tau_i\}$. In particular,
$S_i(t),M_i(t)\geq 1$. For convenience, let
$L_i(t):=S_i(t)+M_i(t)$ and note that $L_i(t)\geq 2$.

It suffices to consider productive interactions, since all other
interactions leave the support vector, and hence $R_i(t)$, unchanged. For
convenience, write $e(t)=\bot$ if the selected interaction is not productive.
Conditional on a productive edge being activated, either endpoint copies
the other with probability $1/2$. We distinguish three cases.

\begin{enumerate}
    \item Suppose that the activated edge is $e(t)=\{i,j\}$ for some
    $j\in N_i$. The pair $(S_i(t),M_i(t))$ then changes to either
    $(S_i(t)+1,M_i(t)-1)$ or $(S_i(t)-1,M_i(t)+1)$, each with probability
    $1/2$. In both cases, $L_i(t)$ remains unchanged. Therefore,
    \begin{align*}
        \E{
            R_i(t+1)
            \growingmid
            \mathcal F(t),e(t)=\{i,j\}
        }
        &=
        \frac{1}{2}
        \frac{S_i(t)+1}{S_i(t)+1+M_i(t)-1}
        +
        \frac{1}{2}
        \frac{S_i(t)-1}{S_i(t)-1+M_i(t)+1}
        \\
        &=
        \frac{S_i(t)}{S_i(t)+M_i(t)}
        =R_i(t).
    \end{align*}

    \item Suppose that the activated edge is $e(t)=\{j,k\}$, where
    $j\in N_i$ and $k\notin N_i\cup\{i\}$. In this case, $S_i(t)$ remains
    unchanged, while $M_i(t)$, and hence $L_i(t)$, increases or decreases
    by one with equal probability. It follows that
    \begin{align*}
        \E{
            R_i(t+1)
            \growingmid
            \mathcal F(t),e(t)=\{j,k\}
        }
        &=
        \frac{S_i(t)}{2}
        \left(
            \frac{1}{L_i(t)-1}
            +
            \frac{1}{L_i(t)+1}
        \right)
        =
        \frac{S_i(t)L_i(t)}{L_i(t)^2-1}
        \\
        &=
        \frac{S_i(t)}{L_i(t)}
        +
        \frac{S_i(t)}
        {L_i(t)\bigl(L_i(t)^2-1\bigr)}
        \geq
        \frac{S_i(t)}{L_i(t)}
        =R_i(t).
    \end{align*}

    \item Finally, suppose that both endpoints of the activated edge belong
    to $N_i$. Although the interaction may redistribute support between
    these two opinions, it does not change their total support and hence
    leaves $M_i(t)$ unchanged. An interaction whose endpoints both lie
    outside $N_i\cup\{i\}$ also leaves $S_i(t)$ and $M_i(t)$ unchanged. In
    either case, $R_i(t+1)=R_i(t)$.
\end{enumerate}

Thus, conditional on $\mathcal F(t)$ and on any possible interaction, the
expected value of $R_i(t+1)$ is at least $R_i(t)$. By the law of total
expectation,
\begin{align*}
    \E{R_i(t+1)\growingmid\mathcal F(t)}
    &=
    \sum_{e\in E\cup\{\bot\}}
    \pr{e(t)=e\growingmid\mathcal F(t)}
    \E{R_i(t+1)\growingmid\mathcal F(t),e(t)=e}
    \\
    &\geq
    \sum_{e\in E\cup\{\bot\}}
    \pr{e(t)=e\growingmid\mathcal F(t)}R_i(t)
    =R_i(t).
\end{align*}
Therefore, $(R_i(t))_{t\geq 0}$ is a bounded submartingale.
\end{proof}

We can now complete the proof of the theorem. Since $T<\infty$ almost surely
and $(R_i(t))_{t\geq 0}$ is bounded, Doob's Optional Stopping Theorem applies
directly. Together with \cref{lem:terminal-value-R,lem:R-submartingale}, this
yields \begin{align*}
    \pr{X_i=1}=\E{X_i}=\E{R_i(T)}\geq R_i(0).
\end{align*}
Finally, by the
definition of $R_i$, we have
\begin{align*}
    R_i(0)=S_i(0)/(S_i(0)+M_i(0))=s_i/(s_i+\sum_{j\in N_i}s_j)=p_i,
\end{align*} which
proves the theorem.

\subsection{Proof of \Cref{thm:survival-tail} (Number of Surviving Opinions)}
\label{sec:survival-tail}
For the proof of the theorem, we need two main ingredients. First, we show that the non-survival events of opinions at pairwise distance at least three satisfy a suitable negative product condition.

\begin{lemma}
\label{lem:product-distance-three}
Let $I\subseteq O$ be a set of opinions at pairwise distance at least three in $G$, and let $Q_i(t):=1-R_i(t)$. Then
$\left(\prod_{i\in I}Q_i(t)\right)_{t\geq 0}$
is a supermartingale.
\end{lemma}
The proof of this lemma is based on elementary but somewhat lengthy one-step calculations, and we postpone it until the end of this section.
Our second ingredient is a variant of the Chernoff bound for random variables whose dependencies are described by a graph. In the standard setting, two variables may be dependent only if their corresponding vertices are adjacent. Consequently, the variables indexed by an independent set are mutually independent. The chromatic number $\chi(H)$ is the smallest number of colors needed to color $H$ such that adjacent vertices receive different colors. Equivalently, it is the smallest number of independent sets into which the vertices of $H$ can be partitioned. Applying the Chernoff argument separately to these color classes leads to a loss of a factor $\chi(H)$ in the exponent; see, for example, \cite{janson2004large}.
More precisely, we show
\begin{restatable}{paperthm}{graphchernoff}
\label{lem:chernoff-local-products}
Let $X_1,\ldots,X_q\in{0,1}$ be binary random variables, let $p_1,\ldots,p_q\in[0,1]$, and let $H$ be an undirected graph on $[q]$. Furthermore, let
$Z:=\sum_{i=1}^qX_i$, $\mu:=\sum_{i=1}^qp_i$, and
$\kappa:=\chi(H)$. Suppose that, for every independent set $I$ of $H$,
\begin{align}
\E{\prod_{i\in I}(1-X_i)}
\leq
\prod_{i\in I}(1-p_i).
\label{eq:local-product-assumption}
\end{align}
Then, for every $0\leq\delta\leq1$,
\begin{align*}
\pr{Z\leq(1-\delta)\mu}
\leq
\exp\left(-\frac{\delta^2\mu}{2\kappa}\right).
\end{align*}
\end{restatable}

The proof follows the standard argument for dependency graphs almost verbatim. In that argument, the vertices are partitioned into $\kappa$ color classes, each of which is an independent set of $H$. Independence within a color class is used only to factor the relevant expectation into a product. In our setting, this step is replaced by the assumption in \eqref{eq:local-product-assumption}. 

We can now derive the theorem from these two statements. We apply \cref{lem:chernoff-local-products} with $H=G^2$, the second power of $G$. Recall that two distinct opinions are adjacent in $G^2$ if their distance in $G$ is at most two. Thus, a set is independent in $G^2$ if and only if its opinions have pairwise distance at least three in $G$. It therefore remains to verify \eqref{eq:local-product-assumption} for every such set. More precisely, we need to show that, for every set $I\subseteq O$ whose opinions have pairwise distance at least three,
\begin{align}
\E{\prod_{i\in I}(1-X_i)}
\leq
\prod_{i\in I}(1-p_i).
\end{align}

This is where \cref{lem:product-distance-three} comes into play. Fix such a set $I$ and consider the bounded supermartingale
$\left(\prod_{i\in I}Q_i(t)\right)_{t\geq 0}$.
Recall that $Q_i(T)=1-X_i$ for every $i\in O$, where $T$ denotes the stabilization time of the process. Applying the Optional Stopping Theorem to $T\wedge t$ and then letting $t\to\infty$, using bounded convergence, gives

\begin{align}
\E{\prod_{i\in I}(1-X_i)}
&=\E{\prod_{i\in I}Q_i(T)}\nonumber\\
&\leq\prod_{i\in I}Q_i(0)
=\prod_{i\in I}(1-p_i).
\end{align}

Thus, the variables $X_i$ satisfy the local product assumption \eqref{eq:local-product-assumption} with respect to $G^2$, and the theorem follows.

It remains to prove \cref{lem:product-distance-three}.

\begin{proof}

Fix $t\geq 0$. If the process has already stabilized, the product remains constant and there is nothing to prove. Otherwise, condition on $\mathcal F(t)$ and on the unordered opinion edge $e(t)$ activated in step $t$. To prove that the product process is a supermartingale, we must show that, for every possible choice of $e(t)$,
\begin{align}
    \E{\prod_{j\in I}Q_j(t+1) \growingmid \mathcal F(t),e(t)}-\prod_{j\in I}Q_j(t)\leq 0.
\end{align}

If $Q_i(t)=0$ for some stopped $i\in I$, then the product remains zero forever and there is nothing to prove. A stopped factor with value $1$ can be removed from the product. Hence, we may assume that all ratios considered below are active.

Since the opinions in $I$ have pairwise distance at least three, their closed neighborhoods are pairwise disjoint. Consequently, an opinion edge can affect at most two of the ratios. Up to exchanging the endpoints of $e(t)$, we distinguish the following cases.

\begin{enumerate}
    \item $e(t)=\{i,k\}$ for some $i\in I$ and $k\in N_i$.

    Among the quantities entering $Q_i(t)$, only $S_i(t)$ changes, while $L_i(t)$ remains unchanged. Therefore,

\begin{align*}
        &\E{\prod_{j\in I}Q_j(t+1) \growingmid \mathcal F(t),e(t)}
            -\prod_{j\in I}Q_j(t) \\
        &\quad= \prod_{j\in I\setminus\{i\}}Q_j(t) \cdot \E{ Q_i(t+1) \growingmid \mathcal F(t),e(t)}
            -\prod_{j\in I}Q_j(t)\\
        &\quad= \prod_{j\in I\setminus\{i\}}Q_j(t)
        \left(
            \frac{1}{2}
            \frac{L_i(t)-S_i(t)+1}{L_i(t)}
            +\frac{1}{2}
            \frac{L_i(t)-S_i(t)-1}{L_i(t)}
            -Q_i(t)
        \right)\\
        &\quad= \prod_{j\in I\setminus\{i\}}Q_j(t)
        \left(
            Q_i(t)-Q_i(t)
        \right)
        =0.
    \end{align*}

    \item $e(t)=\{k,\ell\}$ for some $i\in I$ with $k\in N_i$ and
    $\ell\notin\bigcup_{j\in I}(N_j\cup\{j\})$.

    In this case, $S_i(t)$ remains unchanged, while $L_i(t)$ becomes either
    $L_i(t)-1$ or $L_i(t)+1$, each with probability $1/2$. Since the ratio is active, we have $L_i(t)\geq 2$. Hence,

\begin{align*}
        &\E{\prod_{j\in I}Q_j(t+1) \growingmid \mathcal F(t),e(t)}
            -\prod_{j\in I}Q_j(t) \\
        &\quad= \prod_{j\in I\setminus\{i\}}Q_j(t) \cdot \E{ Q_i(t+1) \growingmid \mathcal F(t),e(t)}
            -\prod_{j\in I}Q_j(t)\\
        &\quad=\prod_{j\in I\setminus\{i\}}Q_j(t)
        \left(
            \frac{1}{2}
            \frac{L_i(t)-1-S_i(t)}{L_i(t)-1}
            +\frac{1}{2}
            \frac{L_i(t)+1-S_i(t)}{L_i(t)+1}
            -Q_i(t)
        \right)\\
        &\quad=\prod_{j\in I\setminus\{i\}}Q_j(t)
        \left(
            1-\frac{L_i(t)S_i(t)}{L_i(t)^2-1}
            -Q_i(t)
        \right)\\
        &\quad=\prod_{j\in I\setminus\{i\}}Q_j(t)
        \left(
            R_i(t)-R_i(t)\left(1-\frac{1}{L_i(t)^2}\right)^{-1}
        \right)\\
        &\quad=-\prod_{j\in I\setminus\{i\}}Q_j(t)
        \frac{R_i(t)}{L_i(t)^2-1}
        \leq 0.
    \end{align*}

    \item $e(t)=\{k,\ell\}$ for distinct $i,j\in I$ with
    $k\in N_i$ and $\ell\in N_j$.

    This case can only occur when $\operatorname{dist}_G(i,j)=3$. First, we rewrite the term. Using the product-increment identity

\begin{align*}
    &Q_i(t+1)Q_j(t+1)-Q_i(t)Q_j(t) \\
        &\quad=
        Q_i(t)\bigl(Q_j(t+1)-Q_j(t)\bigr) +
        Q_j(t)\bigl(Q_i(t+1)-Q_i(t)\bigr)\\
    &\qquad{}+
        \bigl(Q_i(t+1)-Q_i(t)\bigr)
        \bigl(Q_j(t+1)-Q_j(t)\bigr),
\end{align*}

    and using that $Q_i(t)$ and $Q_j(t)$ are $\mathcal F(t)$-measurable, we obtain

\begin{align*}
        &\E{Q_i(t+1)Q_j(t+1)\growingmid\mathcal F(t),e(t)}
            -Q_i(t)Q_j(t) \\
        &\quad=
        Q_i(t)\left(
            \E{Q_j(t+1)\growingmid\mathcal F(t),e(t)}-Q_j(t)
        \right)\\
        &\qquad{}+Q_j(t)\left(
            \E{Q_i(t+1)\growingmid\mathcal F(t),e(t)}-Q_i(t)
        \right)\\
        &\qquad{}+\E{
            \bigl(Q_i(t+1)-Q_i(t)\bigr)
            \bigl(Q_j(t+1)-Q_j(t)\bigr)
            \growingmid\mathcal F(t),e(t)
        }.
    \end{align*}

    We now show that each of the three summands is non-positive. The first two terms are non-positive by the calculations from the previous case. Indeed,

\begin{align*}
        &\E{Q_i(t+1)\growingmid\mathcal F(t),e(t)}-Q_i(t) \\
        &\quad=
        \frac{1}{2}
        \left(
            \frac{L_i(t)-1-S_i(t)}{L_i(t)-1}
            +
            \frac{L_i(t)+1-S_i(t)}{L_i(t)+1}
        \right)
        -\frac{L_i(t)-S_i(t)}{L_i(t)}
        \leq 0.
    \end{align*}

    The same holds analogously for $j$.

    It remains to consider the last term. Define $\iota_{ij}\in\{-1,1\}$ as the direction of the interaction, where $\iota_{ij}=1$ if an agent moves from $\ell\in N_j$ to $k\in N_i$, and $\iota_{ij}=-1$ otherwise. Thus,
    \begin{align*}
        L_i(t+1)&=L_i(t)+\iota_{ij},&
        L_j(t+1)&=L_j(t)-\iota_{ij}.
    \end{align*}
    Consequently,

\begin{align*}
        Q_i(t+1)-Q_i(t)
        &=
        \left(1-\frac{S_i(t)}{L_i(t)+\iota_{ij}}\right)
        -\left(1-\frac{S_i(t)}{L_i(t)}\right) \\
        &= \frac{\iota_{ij}S_i(t)}
        {L_i(t)(L_i(t)+\iota_{ij})},
    \end{align*}

    and, analogously,

\begin{align*}
        Q_j(t+1)-Q_j(t)
        &=
        \left(1-\frac{S_j(t)}{L_j(t)-\iota_{ij}}\right)
        -\left(1-\frac{S_j(t)}{L_j(t)}\right) \\
        &= -\frac{\iota_{ij}S_j(t)}
        {L_j(t)(L_j(t)-\iota_{ij})}.
    \end{align*}

    Therefore,

\begin{align*}
    \mathcal{Q}_{ij}(t)
        &:=
        \bigl(Q_i(t+1)-Q_i(t)\bigr)
        \bigl(Q_j(t+1)-Q_j(t)\bigr) \\
        &= -\frac{S_i(t)S_j(t)}
        {L_i(t)(L_i(t)+\iota_{ij})L_j(t)(L_j(t)-\iota_{ij})}.
\end{align*}

    Since both outcomes of $\iota_{ij}$ are equally likely,

\begin{align*}
        &\E{\mathcal{Q}_{ij}(t)
            \growingmid\mathcal F(t),e(t)} \\
        &\quad= \frac{1}{2}
        \left(
            -\frac{S_i(t)S_j(t)}
            {L_i(t)L_j(t)(L_i(t)-1)(L_j(t)+1)}
            -
            \frac{S_i(t)S_j(t)}
            {L_i(t)L_j(t)(L_i(t)+1)(L_j(t)-1)}
        \right)\\
        &\quad=
        -\frac{S_i(t)S_j(t)}{2L_i(t)L_j(t)}
        \left(
            \frac{1}{(L_i(t)-1)(L_j(t)+1)}
            +
            \frac{1}{(L_i(t)+1)(L_j(t)-1)}
        \right)
        \leq 0.
    \end{align*}

    Finally, all other factors remain unchanged, and hence

\begin{align*}
        &\E{\prod_{h\in I}Q_h(t+1)\growingmid\mathcal F(t),e(t)}
            -\prod_{h\in I}Q_h(t) \\
        &\quad=
        \prod_{h\in I\setminus\{i,j\}}Q_h(t)
        \left(
            \E{Q_i(t+1)Q_j(t+1)\growingmid\mathcal F(t),e(t)}
            -Q_i(t)Q_j(t)
        \right)
        \leq 0.
    \end{align*}

    \item In all remaining cases, $e(t)$ does not change any $Q_i(t)$ with $i\in I$, and hence
    \begin{align*}
        \E{\prod_{i\in I}Q_i(t+1)\growingmid\mathcal F(t),e(t)}
        -\prod_{i\in I}Q_i(t)
        =0.
    \end{align*}
\end{enumerate}

For an idle step, $e(t)=\bot$ and the product remains unchanged.
Since the conditional drift is non-positive for every possible choice of $e(t)$, the law of total expectation gives

\begin{align*}
    &\E{\prod_{i\in I}Q_i(t+1)\growingmid\mathcal F(t)} \\
    &\quad=
    \sum_{e\in E\cup\{\bot\}}
    \pr{e(t)=e\mid\mathcal F(t)}
    \E{\prod_{i\in I}Q_i(t+1)
        \growingmid\mathcal F(t),e(t)=e} \\
    &\quad\leq \prod_{i\in I}Q_i(t),
\end{align*}

which proves the lemma.
\end{proof}

\subsection{Proof of \Cref{prop:cycleshatter} (Maximal Support on Cycle Graphs)}
\label{sec:a ring to rule them all}

We first explain the analysis approach of our main results in \Cref{prop:shatteringdos} for a simple cycle graph and a fixed starting configuration. This serves as a compact illustration and a warm-up for the complete technical analysis for bounded degree graphs in Section~\ref{sec:bounded degree bound} below.
In this section, we consider the productive \EdgeVOG process on $G_r=(O_r,E_r)$ with $O_r= \lbrace 1,...,m \rbrace$ and $\lbrace i, j \rbrace \in E_r$ whenever $j = 1+(i\bmod m)$ for all $i \in O_r$. There are $n$ agents, where the initial support of each opinion $i \in O_r$ is 1, i.e., $n = m$ and $|S_i{(0)}| = 1$ for all $i$.
In the following, we adopt the random walk-based view of the process where the agents are assigned to the node of the opinion graph corresponding to the opinion they have. On a high level, the main argument for the cycle graph is as follows. Suppose that, at some point,  a large number of agents is assigned to a specific node $i \in O$. Now consider one of these agents and assume that the agent followed the path $p$ to $i$. It is easy to see that each node on $p$ was active at least twice (once in the step that $i$ entered the node and once in the step that $i$ left the node).  However, as we will see later, each node is, on average,  active only rarely before being decided. Hence it is unlikely that there are many long paths crossing node $i$.

Fix an opinion $i\in O$ and a constant $c\geq1$, and let
$d\coloneqq\lceil2(c+2)\log_2 n\rceil$.
If $2d+1\geq n$, the desired bound follows directly from $S_i(t)\leq n$.
Otherwise, let $C_d$ be the path consisting of $i$ and the $d$ closest opinions on either side of $i$ in the cycle. Initially, this path contains exactly $2d+1$ agents.
Thus, if $i$ ever has more than $2d+1$ agents, at least one of them must have started outside $C_d$.
Such an agent must traverse one of the two paths of $d$ opinions leading to $i$. Each of these $d$ opinions is involved in at least two productive interactions: one in which the agent enters and one in which it leaves.

In the following, we focus, w.l.o.g., on the left path.
Choose a fixed independent set $I$ of at least $d/2$ opinions on this path.
For each $j\in I$, consider its first productive interaction, if such an interaction occurs.
Since $S_j(0)=1$ and $j$ has not previously been productive, its support immediately before this interaction is still one. With probability $1/2$, its sole agent leaves, so $j$ becomes extinct and can never be productive again.
Consequently, if $j$ is productive at least twice, its first productive interaction must increase its support.

To account for the dependence between the opinions, assign an independent fair coin to each $j\in I$ and use this coin to choose the direction of its first productive interaction. Since $I$ is independent, an interaction involves at most one opinion in $I$, so this gives exactly the original process. The remaining interactions use fresh fair coins. If a first interaction never occurs, its assigned coin is simply unused.
Thus, traversing the path requires a prescribed outcome of every coin assigned to $I$, and its probability is at most
\[
    2^{-|I|}\leq 2^{-d/2}.
\]
The same argument holds for the right path. Applying the union bound over both directions and all $n$ opinions gives
\[
    \pr{\exists i\in O,\ t\geq0:S_i(t)>2d+1}
    \leq 2n\cdot2^{-d/2}
    \leq 2n^{-(c+1)}\leq n^{-c}.
\]
As $d\in O(\log n)$, this proves the bound simultaneously for all opinions and all time steps. The ordinary process visits the same configurations after idle steps are removed, so the result also holds for that process.

\subsection{Proof of \Cref{prop:shatteringdos} (Maximal Support on Bounded Degree Graphs)}
\label{sec:bounded degree bound}
\newcommand{\indicator}{\ensuremath{\mathds{1}_{\{ P_j(\ato T) \geq \tau'\}}}}

\label{sec:support}
In this section we prove \Cref{prop:shatteringdos}. It bounds the maximal support of an opinion as a function of the maximal initial support and the maximal degree in the opinion graph.
In particular, we show the following more detailed theorem
\begin{theorem}
\label{prop:shatteringdosX}

Let $G=(O,E)$ be an arbitrary graph with maximum degree at most $\Delta$, where $\Delta\geq1$. Suppose $\hat s$ is the maximal initial support.
For both the \EdgeVOG and the productive \EdgeVOG processes, with high probability every opinion has support at most
$O\left(\min\{n,\Delta^6\hat s\log^2n,\Delta^9\hat s^2\log n\}\right)$ simultaneously for all steps $t\geq0$.

\end{theorem}

\noindent \emph{Proof overview:}
We generalize the arguments for the cycle in Section~\ref{sec:a ring to rule them all} above. For a given opinion $i \in O$ we want to bound the subset of opinions whose agents might end up on $i$ in the final configuration. We need a more general notion than the $O(\log n)$-neighborhood, since it may already contain all opinions and agents. Moreover, multiple agents may start on the same opinion, and therefore each such opinion is expected to be productive more than once before becoming decided.

Instead, our main tool for the analysis are \emph{witness sets} defined as follows: For each opinion $i \in O$, time $t \ge 0$, and threshold $\tau \ge 0$, the witness set $W_i^\tau(t)$ (cf.\ \Cref{def:witness} below) is the largest connected subgraph containing $i$ in which every opinion has been productive at least $\tau$ times by time~$t$. Thus, the witness set depends on both the topology of the opinion graph and the configuration in step $t$.
We furthermore assign each witness a weight proportional to its size, scaled by $\hat{s}+\min\{\Delta\tau,\sqrt{\Delta\tau}\log n\}$ (cf. \Cref{def:witnessweight} below). This weight gives a coarse upper bound on how many agents $W_i^{\tau}(t)$ can contain in step $t$.

The proof proceeds in two steps:
First, in \Cref{lemma:upperboundS}, we show that bounding the weight of $W_i^\tau(t)$ directly bounds the support $S_i(t)$, w.h.p.
As the weight depends on the parameter $\tau$ and the corresponding size $|W^{\tau}_i(t)|$, we need to argue that there is a small $\tau$ where the witness sets remain small with high probability.
We will show this in \Cref{lemma:Ctau}.
The main complication in proving the propositon is that productive steps of neighboring opinions are strongly correlated.
For example, if a given opinion was productive $\tau$ times, then, in total, all its neighbors must also have been productive $\tau$ times.
Thus, there is a positive correlation between productive events of adjacent opinions.
To deal with this dependence,
in \Cref{lemma:pkbound}, we show that productive steps of opinions at distance at least four behave independently for the purpose of tail bounds

We then use this insight in \Cref{lemma:Ctau} to prove that for $\tau'=100\Delta^5\hat s^2$, witness sets cannot grow larger than $O(\Delta^3\log n)$ with high probability.

Together with \Cref{lemma:upperboundS}, this implies the desired bound on $S_i(t)$, completing the proof.

\noindent Let us start with a formal definition of the witness sets.

For each opinion $i$ and each step $t$, we define $P_i(\ato t)$ as the number of productive interactions involving $i$ in steps $0,\ldots,t-1$. Thus, $P_i(\ato0)=0$, and these counts refer to the configuration at the beginning of step $t$.

Fix a node $i\in V$, a step $t\geq 0$ and a parameter $\tau\ge 1$.
Then, we define the \emph{$\tau$-witness} $W^\tau_i(t)$ of $i$  in step $t$ as the largest connected subset containing $i$
in which every node experiences at least $\tau$ productive steps until step $t$.
\begin{definition}[$\tau$-Witness]
\label{def:witness}
Let $i \in V$ be an arbitrary opinion and let $t \geq 0$ be a step.
For any integer $\tau\geq0$, we define the set $O_\tau(t) = \{ j \in O \mid P_j(\ato{t}) \geq \tau\}$ and let $G_{\tau}(t)$ be the subgraph induced by $O_\tau(t)$.
Then, define the unique \emph{$\tau$-witness} $W^\tau_i(t) \subset O$ of $i$ as the connected component containing $i$ in $G_{\tau}(t)$ if $i \in O_\tau(t)$ and $\{i\}$ otherwise.
\end{definition}
\noindent Observe that for different choices of $\tau$, we obtain widely different witness sets in the same configuration.
Furthermore, we define the \emph{weight} of a witness set.
\begin{definition}[Weight of a $\tau$-Witness]
\label{def:witnessweight}
Suppose \EdgeVOG starts in a configuration with maximum support $\hat{s}$. Then, for a fixed $c \geq 1$, we define the \emph{weight} of $W^\tau_i(t)$ for opinion $i \in O$ in step $t$ and integer $\tau\geq0$, as
\begin{align*}
\omega_c(W^\tau_i(t))\coloneqq |W^\tau_i(t)| \cdot \left(\hat{s}+\min\left\{\Delta\tau, 6c\sqrt{\Delta\tau} \log n\right\}\right).
\end{align*}
\end{definition}
The weight $\omega_c(W_i^\tau(t))$ represents a coarse upper bound on how many agents can be contained in $W_i^\tau(t)$ by time $t$. Intuitively, the argument for this is as follows:

Clearly, the agents initially in $W_i^\tau(t)$ contribute at most $|W_i^\tau(t)|\hat s$ to its support. Any increase in this support is the net traffic across its boundary.
If $i$ has been productive at least $\tau$ times, every opinion on the outer boundary has been productive fewer than $\tau$ times, by maximality of the witness. If $i$ has been productive fewer than $\tau$ times, then the witness is $\{i\}$ and its total boundary count is itself less than $\tau$.
In either case, the number of productive boundary interactions is at most $\Delta|W_i^\tau(t)|\tau$.
Each such interaction sends one agent in either direction with equal probability. In \Cref{lemma:trafficbound}, we bound this net traffic simultaneously over time. Taking a union bound over all possible sets then allows us to apply the estimate to the random witness $W_i^\tau(t)$.

Consequently, for any threshold $\tau$, the weight provides an upper bound on the agents in the subset and therefore also on the support of opinion $i$.

More formally, we show the following bound.
\begin{proposition}
    \label{lemma:upperboundS}

For any $c\geq1$, it holds
\[
    \pr{\forall i\in O,\ t\geq0:\quad
    S_i(t)\leq\min_{\tau\in\mathbb Z_{\geq0}}\omega_c(W_i^\tau(t))}
    \geq1-\frac{1}{n^{c-3}}.
\]

\end{proposition}
\begin{proof}
Fix an opinion $i^* \in O$.
Clearly, the support of $i^*$ is at most the support of \emph{every} connected subset it is contained in.
In particular, it must be at most the support of any witness set.
Thus, for every $t \geq 0$ and \emph{every} integer $\tau\geq0$, it holds that
$S_{i^*}(t) \leq S_{W_{i^*}^\tau(t)}(t).$
Therefore, to prove the proposition, we show that
\begin{align}
    \pr{\forall i^*\in O,\ t,\tau\in\mathbb Z_{\geq0}:\quad S_{W_{i^*}^\tau(t)}(t) \leq \omega_c(W_{i^*}^\tau(t))} \geq 1-\frac{1}{n^{c-3}}.\label{eq:setweight}
\end{align}
We begin the proof with some definitions:
For each set $C \subset V$, let $N_C \coloneqq \{ j \in V \setminus C \mid
\exists \{i, j\} \in E \text{ with } i \in C \}$ be the set of its neighbors.
Further, for each edge $e \in E$ and each step $t \geq 0$, we let $P_e(t)$ be the indicator that edge $e \coloneqq \{i,j\}$ is productive in step $t$, i.e., an agent switches its opinion from $i$ to $j$ or vice versa.
For a set of edges $\mathcal{S} \subseteq E$ define
\begin{align*}
    P_\mathcal{S}(t) = \sum_{e \in \mathcal{S}} P_e(t) && \text{and }&& P_\mathcal{S}(\ato t) \coloneqq \sum_{\tau=0}^{t-1} P_\mathcal{S}(\tau).
\end{align*}
For an ordered pair $(i,j) \in V^2$ with $\{i,j\} \in E$ let $P_{i \to j}(t)$ indicate that an agent moved from $i$ to $j$ in $t$.
Note that
\begin{align*}
    P_e(t) \coloneqq P_{i \to j}(t)+ P_{j \to i}(t).
\end{align*}
Moreover, let $P_{i \to j}(\ato{t}) \coloneqq \sum_{s=0}^{t-1} P_{i \to j}(s)$ count how often an agent moved from $i$ to $j$ before step $t$.
The net \emph{traffic} that went along the edge is then defined as
\begin{align*}
    \mathcal{T}_{i \to j}(\ato t) \coloneqq P_{i \to j}(\ato{t}) - P_{j \to i}(\ato{t}).
\end{align*}
For a subset $C$, define the traffic into $C$ over the cut defined by this set as
\begin{align*}
    \mathcal{T}_{C}(\ato t) \coloneqq \sum_{i \in C} \sum_{j \in N_i \setminus C} \mathcal{T}_{j \to i}(\ato t).
\end{align*}
Given these definitions and assuming the \EdgeVOG process starts in an initial configuration $\Vec{s} = (s_1, \ldots, s_q)$, for any set $C \subseteq V$, we can show the following identity:
\begin{align*}
    S_C(t) &= \sum_{i \in C} S_i(t) = \sum_{i \in C} \left( s_i + \left(\sum_{j \in N_i} P_{j \to i}(\ato{t}) -P_{i \to j}(\ato{t})\right) \right)\\
    &= \sum_{i \in C} s_i + \sum_{i \in C}\left(\sum_{j \in N_i} P_{j \to i}(\ato{t}) - P_{i \to j}(\ato{t})\right)\\
    &= \sum_{i \in C} s_i + \sum_{i \in C}\left(\sum_{j \in N_i} \mathcal{T}_{j \to i}(\ato t)\right)
\\
    &= \sum_{i \in C} s_i + \sum_{\{i,j\} \in E[C]} \left(\mathcal{T}_{i \to j}(\ato t) + \mathcal{T}_{j \to
     i}(\ato t) \right) + \sum_{i \in C}\left(\sum_{j \in N_i\setminus C} \mathcal{T}_{j \to i}(\ato t)\right).
\end{align*}
Note that for each $\{i, j\} \in E$ with $i,j \in C$, holds:
\begin{align*}
\mathcal{T}_{i \to j}(\ato t) + \mathcal{T}_{j \to
     i}(\ato t) = \left( P_{i \to j}(\ato{t})  - P_{j \to i}(\ato{t})  \right) + \left( P_{j \to i}(\ato{t}) -P_{i \to j}(\ato{t})  \right) = 0.
\end{align*}
Therefore, the formula simplifies to
\begin{align}
    S_C(t)
     = \sum_{i \in C} s_i + \sum_{i \in C}\left(\sum_{j \in N_i\setminus C} \mathcal{T}_{j \to i}(\ato t)\right) = s_C + \mathcal{T}_{C}(\ato t). \label{eq:sdef}
\end{align}

For the next step, we will bound the traffic $\mathcal T_C(\ato t)$.
Let $\mathcal S(C)$ denote the edges of the cut between $C$ and $O\setminus C$.
In \Cref{lemma:trafficbound} below, we show that for every nonempty fixed set $C$ and every $c\geq1$,
\begin{align}
    \pr{\exists t\geq0:\mathcal T_C(\ato t)
    \geq3c\left(\sqrt{|C|P_{\mathcal S(C)}(\ato t)}+|C|\right)\log n}
    \leq n^{-c|C|}.
    \label{eq:traffic}
\end{align}
Equipped with this bound, consider the event
\[
    \mathcal E=\bigcap_{\emptyset\neq C\subseteq O}
    \bigcap_{t\geq0}
    \left\{\mathcal T_C(\ato t)\leq
    3c\left(\sqrt{|C|P_{\mathcal S(C)}(\ato t)}+|C|\right)\log n\right\}.
\]
For $c\leq3$, the proposition is immediate, so assume $c>3$.
Since initially empty opinions have been discarded, $q\leq n$. Thus, the union bound gives
\begin{align*}
    \pr{\overline{\mathcal E}}
    &\leq\sum_{k=1}^q\binom qk n^{-ck}
    \leq\sum_{k=1}^n n^{-(c-1)k}
    \leq n^{-(c-2)}\leq n^{-(c-3)}.
\end{align*}
In particular, this event holds simultaneously for every set and time, so it applies to the random witness sets without further conditioning.

Fix $i^*$, $t$ and an integer $\tau\geq0$, and write $W=W_{i^*}^\tau(t)$.
If $P_{i^*}(\ato t)\geq\tau$, then $W$ is a maximal connected component of opinions whose productive counts are at least $\tau$. Each neighbor outside $W$ therefore has count less than $\tau$.
For every cut edge $e=\{i,j\}$ with $i\in W$ and $j\notin W$, we obtain
\begin{align}
    P_e(\ato t)\leq\min\{P_i(\ato t),P_j(\ato t)\}\leq\tau.
    \label{eq:witness-cut-count}
\end{align}
There are at most $\Delta|W|$ cut edges, so
$P_{\mathcal S(W)}(\ato t)\leq\Delta|W|\tau$.
If instead $P_{i^*}(\ato t)<\tau$, then $W=\{i^*\}$ and
$P_{\mathcal S(W)}(\ato t)=P_{i^*}(\ato t)<\tau\leq\Delta|W|\tau$.
Thus, the same bound holds in both cases.

For $\tau=0$, the cut count and its traffic are zero.
For $\tau\geq1$, on $\mathcal E$ and using $\Delta\geq1$, we have
\begin{align*}
    \mathcal T_W(\ato t)
    &\leq3c|W|\left(\sqrt{\Delta\tau}+1\right)\log n
    \leq6c|W|\sqrt{\Delta\tau}\log n.
\end{align*}
On the other hand, the net traffic is at most the total cut count, giving
\[
    \mathcal T_W(\ato t)\leq\Delta|W|\tau.
\]
Combining both estimates with \eqref{eq:sdef} yields
\begin{align*}
    S_W(t)
    &=s_W+\mathcal T_W(\ato t) \leq |W|\hat s+|W|\min\{\Delta\tau,6c\sqrt{\Delta\tau}\log n\}
    =\omega_c(W).
\end{align*}
As this holds for every opinion, time and threshold on the same event $\mathcal E$, \Cref{eq:setweight} and the proposition follow.

\end{proof}

\begin{lemma}
\label{lemma:trafficbound}

Let $\emptyset\neq C\subseteq O$, and let $\mathcal S(C)$ be the edges of the cut between $C$ and $O\setminus C$.
Then, for every $c\geq1$,
\[
    \pr{\exists t\geq0:\mathcal T_C(\ato t)
    \geq3c\left(\sqrt{|C|P_{\mathcal S(C)}(\ato t)}+|C|\right)\log n}
    \leq n^{-c|C|}.
\]

\end{lemma}
\begin{proof}

In each step, either no cut edge is productive, or the traffic increases or decreases by one with equal probability.
Thus, $(\mathcal T_C(\ato t))_{t\geq0}$ is a martingale with respect to the natural filtration, starts at zero, and has increments in $\{-1,0,1\}$.
Moreover, a nonzero increment occurs precisely at a productive cut interaction, so
\[
    \sum_{s=1}^t\bigl(\mathcal T_C(\ato s)-\mathcal T_C(\ato{s-1})\bigr)^2=P_{\mathcal S(C)}(\ato t).
\]
Applying \Cref{lemma:nonzero-increments} gives, for every deterministic $a>0$ and integer $m\geq1$,
\begin{align}
    \pr{\exists t\geq0:P_{\mathcal S(C)}(\ato t)\leq m,\ \mathcal T_C(\ato t)\geq a}
    \leq e^{-a^2/(2m)}.
    \label{eq:traffic-crossing}
\end{align}

For each integer $1\leq m\leq n^2$, apply this estimate with
$a=3c(\sqrt{|C|m}+|C|)\log n$.
The failure probability is at most
$\exp(-\tfrac92c^2|C|\log^2 n)$.
If $P_{\mathcal S(C)}(\ato t)=0$, the traffic is zero and cannot violate the claimed bound.
If $P_{\mathcal S(C)}(\ato t)>n^2$, the threshold exceeds $n$, whereas
$\mathcal T_C(\ato t)=S_C(t)-s_C\leq n$.
Thus, a union bound over $m=1,\ldots,n^2$ gives
\begin{align*}
    \pr{\exists t\geq0:\mathcal T_C(\ato t)
       \geq3c\left(\sqrt{|C|P_{\mathcal S(C)}(\ato t)}+|C|\right)\log n}
    &\leq n^2\exp\left(-\tfrac92c^2|C|\log^2n\right)
    \leq n^{-c|C|}.
\end{align*}
For the last inequality, we used $n\geq2$, $c,|C|\geq1$ and
$\tfrac92c^2|C|\log n\geq c|C|+2$.
This proves the lemma for both the ordinary and the productive process.

\end{proof}

It remains to find a suitable $\tau$ where the emerging components have small weight, w.h.p. Note that the weight is not simply linear in $\tau$. To see this, consider the following two cases.
If we choose a large $\tau$, the witness set will be small, but weight is still large due to the large $\tau$.
As an extreme example, for $\tau \in \Theta(n^2 \log n)$ it easy to show that at most two nodes are productive this often. Thus, the witness set contains at most $2$ nodes but has a weight of $\Theta(n^2 \log n)$.
On the other hand, a $\tau$ that is too small will result in a large witness set and therefore, a large weight.
Again, as an extreme example, for $\tau = 0$ the witness set is trivially $V$ (assuming that $G$ is connected). Thus, the weight is $\Theta(n)$.

In \Cref{lemma:Ctau}, we show a much smaller bound for $\tau' \in \Theta{(\Delta^5 \hat{s}^2)}$. Let $T$ be the step in which the process converges. Then it holds that $W^{\tau'}_i(T) \leq O(\Delta^3 \log n)$ with high probability.
\begin{proposition}
\label{lemma:Ctau}

Let $\tau'=100\Delta^5\hat s^2$. Then, there is a universal constant $C\geq1$ such that, for every opinion $i$ and every integer $k\geq1$,
\[
    \pr{\exists t\geq0:|W_i^{\tau'}(t)|\geq C\Delta^3 k
          \mid\vec S(0)=\vec s}\leq2^{-k}.
\]

\end{proposition}

\noindent Before we prove the proposition, we show how it combines with \Cref{lemma:upperboundS} to imply \Cref{prop:shatteringdosX}.
Fix a constant $c>4$ and define
\[
    f_c(\Delta,\hat s)
    \coloneqq\hat s+\min\{100\Delta^6\hat s^2,
                           60c\Delta^3\hat s\log n\}.
\]
Substituting $\tau'=100\Delta^5\hat s^2$ in the definition of the weight gives exactly
\[
    \omega_c(W_i^{\tau'}(t))=|W_i^{\tau'}(t)|f_c(\Delta,\hat s).
\]
Let $k=\lceil(c+2)\log_2 n\rceil$, and consider the events
\begin{align*}
    \mathcal E_1&=\{\forall i,t,\tau:\ S_i(t)\leq\omega_c(W_i^\tau(t))\},\\
    \mathcal E_2&=\{\forall i,t:\ |W_i^{\tau'}(t)|<C\Delta^3 k\}.
\end{align*}
By \Cref{lemma:upperboundS}, $\pr{\overline{\mathcal E_1}}\leq n^{-(c-3)}$.
By \Cref{lemma:Ctau} and a union bound over at most $n$ opinions,
$\pr{\overline{\mathcal E_2}}\leq n2^{-k}\leq n^{-(c+1)}$.
On $\mathcal E_1\cap\mathcal E_2$, every support at every time is at most $C\Delta^3 kf_c(\Delta,\hat s)$.
Consequently,
\begin{align}
    \pr{\forall i\in O,\ t\geq0:\ S_i(t)\leq C\Delta^3 kf_c(\Delta,\hat s)}
    &\geq\pr{\mathcal E_1\cap\mathcal E_2} \geq1-n^{-(c-3)}-n^{-(c+1)}.
    \label{eq:helpereq}
\end{align}
For fixed $c$, the support threshold is
$O(\min\{\Delta^9\hat s^2\log n,\Delta^6\hat s\log^2n\})$.
The additional term $\Delta^3\hat s\log n$ is absorbed since $\Delta,\hat s\geq1$ and $n\geq2$.
Combining this with the deterministic bound $S_i(t)\leq n$ proves \Cref{prop:shatteringdosX}. By choosing $c$ sufficiently large, the failure exponent can be made arbitrarily large.
\medskip

\noindent In the remainder of this section,  we prove \Cref{lemma:Ctau}.
In particular, we show that, w.h.p., it holds that $|W^{\tau'}_i(T)| \leq O(\Delta^3 \log n)$ where $T$ is the step where the process converges.
Note that this is a bound on $W^\tau_i(t)$ for all values of $t \geq 0$.
On the hand, note that $W^{\tau'}_i(t)$ is monotonically increasing in $t$.
For any connected set, in which each node is productive $\tau'$ times until step $t$, this must clearly also hold in any step $t' \geq t$.
Thus, $|W^{\tau'}_i(T)|$ is an upper bound for all $|W^{\tau'}_i(t)|$ with $t \leq T$.
On the other hand, in step $t \geq T$, no support changes because the process has converged. Therefore, $|W^{\tau'}_i(T)| = |W^{\tau'}_i(t)|$ for all $t \geq T$.
Therefore, it suffices to bound $|W^{\tau'}_i(T)|$ to derive an upper bound for all $t \geq 0$.

In order to prove \Cref{lemma:Ctau}, we need some additional definitions.

We call an opinion $j\in O$ \emph{bad} if its total number $P_j(\ato T)$ of productive interactions before stabilization is at least $\tau'=100\Delta^5\hat s^2$. This count includes the interaction in which $j$ becomes decided.

Otherwise, we call it \emph{good}.
Likewise, we call a set $S \subseteq V$ bad if \emph{every} $i \in S$ is bad, and good otherwise.
After $j$ becomes decided, it is never productive again.
In equations, we will use the indicator variable $\indicator$ to denote whether an opinion is good or not.
In particular, we let $\indicator=1$ if opinion $j$ is bad, and $\indicator=0$ otherwise.

Given these definitions, we want to show that for each fixed opinion $i$, the probability that a connected set of bad nodes containing $i$ has size at least $C\Delta^3 k$ is at most $2^{-k}$, for a suitable constant $C$.
A major complication in showing that a subgraph $C \subset V$ is bad is that the events that two opinions $j,j' \in C$ are bad are \emph{not} necessarily independent.
In other words, we cannot directly bound each opinion $j \in C$ by applying Chernoff bounds individually. Instead of considering the connected sets directly, we will consider sets where opinions are in distance at least $4$.
We call them \emph{$4$-separated sets}. For these sets, the events that the individual nodes are bad satisfy a joint upper bound.

In \Cref{lemma:pkbound} below, we show that for every $4$-separated set $\mathcal I$ and every $\tau'\geq1$,
\[
    \pr{\forall j\in\mathcal I:P_j(\ato T)\geq\tau'}
    \leq\left(\frac{\Delta\hat s^2}{\tau'}\right)^{|\mathcal I|}.
\]

\begin{lemma}
    \label{lemma:pkbound}

Consider a $4$-separated set $\mathcal I$. Then, for every $\tau'\geq1$,
$\pr{\forall j\in\mathcal I:P_j(\ato T)\geq\tau'}\leq(\Delta\hat s^2/\tau')^{|\mathcal I|}$.

\end{lemma}
\begin{proof}

Before we give the idea and structure of this proof, we require some helpful definitions.
W.l.o.g. let $\mathcal I=\{1,\ldots,k\}$, and let $(\mathcal F(t))_{t\geq0}$ be the natural filtration generated by all edge selections and coin flips.
Recall that $P_i(\ato t)$ counts the productive interactions involving $i$ before time $t$. In particular, $P_i(\ato0)=0$.
Consider
\begin{align*}
    \Phi(t)\coloneqq\prod_{i\in\mathcal I}\bigl(\Psi_i(\Vec S(t))+P_i(\ato t)\bigr).
\end{align*}
We first show the following claim.
\begin{claim}
\label{claim:phiismartingale}
The process $(\Phi(t))_{t\geq0}$ is a nonnegative martingale.
\end{claim}
\begin{proof}
To establish that $\Phi(t)$ is a martingale, we analyze how it behaves depending on the random productive edge $e(t)$ chosen in step $t$.
Fix a time $t$ and condition on $\mathcal F(t)$. We write $e(t)=\bot$ for an idle step and use $\Psi_i(t)$ as shorthand for $\Psi_i(\Vec S(t))$.
Recall from \Cref{claim:casedist} that, conditional on the chosen productive edge, its potential decreases by one in expectation and all other edge potentials remain unchanged in expectation.
Thus, the node potential of a productive opinion decreases by one in expectation, while the other node potentials remain unchanged in expectation.
Moreover, since $\mathcal I$ is $4$-separated, an edge cannot meet the closed neighborhoods of two different opinions in $\mathcal I$, as this would give a path of length at most three between them. Hence, at most one factor of $\Phi(t)$ can change in a step.
With this in mind, we consider the following case distinction:

\begin{itemize}[itemsep=2pt,topsep=2pt,parsep=0pt,partopsep=0pt,leftmargin=12pt]
    \item If one endpoint of the chosen edge $e$ lies in $\mathcal I$, say $i\in\mathcal I$, then $P_i(\ato{t+1})=P_i(\ato t)+1$.
    By \Cref{claim:casedist}, we have
    \[
        \E{\Psi_i(t+1)\mid\mathcal F(t),e(t)=e}=\Psi_i(t)-1,
    \]
    while all factors belonging to opinions in $\mathcal I\setminus\{i\}$ remain unchanged.
    Thus, we get
    \begin{align*}
        \E{\Phi(t+1)\mid\mathcal F(t),e(t)=e}
        &=\left(\prod_{u\in\mathcal I\setminus\{i\}}(\Psi_u(t)+P_u(\ato t))\right)
          \E{\Psi_i(t+1)+P_i(\ato t)+1\mid\mathcal F(t),e(t)=e}\\
        &=\left(\prod_{u\in\mathcal I\setminus\{i\}}(\Psi_u(t)+P_u(\ato t))\right)
          (\Psi_i(t)-1+P_i(\ato t)+1)=\Phi(t).
    \end{align*}

    \item If neither endpoint of $e$ lies in $\mathcal I$, but one has a neighbor $i\in\mathcal I$, then all counts $P_j(\ato t)$ with $j\in\mathcal I$ remain unchanged because no opinion in $\mathcal I$ is productive.
    There cannot be a second opinion in $\mathcal I$ whose node potential is affected, by the distance argument above.
    Since $e$ is not incident to $i$, every edge potential contributing to $\Psi_i$ remains unchanged in expectation by \Cref{claim:casedist}. Consequently,
    \[
        \E{\Psi_i(t+1)\mid\mathcal F(t),e(t)=e}=\Psi_i(t).
    \]
    All other factors remain unchanged, so
    \begin{align*}
        \E{\Phi(t+1)\mid\mathcal F(t),e(t)=e}
        &=\left(\prod_{u\in\mathcal I\setminus\{i\}}(\Psi_u(t)+P_u(\ato t))\right)
          \E{\Psi_i(t+1)+P_i(\ato t)\mid\mathcal F(t),e(t)=e}\\
        &=\left(\prod_{u\in\mathcal I\setminus\{i\}}(\Psi_u(t)+P_u(\ato t))\right)
          (\Psi_i(t)+P_i(\ato t))=\Phi(t).
    \end{align*}

    \item Finally, if neither endpoint lies in $\mathcal I$ nor is adjacent to an opinion in $\mathcal I$, then no potential and no count in the product changes. The same holds for an idle step. Hence, $\Phi(t+1)=\Phi(t)$ deterministically in either case.
\end{itemize}
As these cases cover all possible outcomes, the law of total expectation gives
\begin{align*}
    \E{\Phi(t+1)\mid\mathcal F(t)}
    &=\sum_{e\in E\cup\{\bot\}}\pr{e(t)=e\mid\mathcal F(t)}
      \E{\Phi(t+1)\mid\mathcal F(t),e(t)=e}=\Phi(t).
\end{align*}
Finally, $\Phi(t)$ is determined by the history up to time $t$, all its factors are nonnegative, and $\Phi(t)\leq(n^2+t)^k$ at every fixed time $t$. Thus, it is adapted and integrable, and $(\Phi(t))_{t\geq0}$ is indeed a nonnegative martingale.
\end{proof}

Equipped with this fact, we can now bound the probability that all opinions in $\mathcal I$ are bad.
First observe that
\begin{align}
    \Phi(0)
    &=\prod_{i\in\mathcal I}s_i\sum_{j\in N_i}s_j
      \leq(\Delta\hat s^2)^k.
    \label{eq:psilower}
\end{align}
At stabilization, all node potentials are zero. Therefore,
$\Phi(T)=\prod_{i\in\mathcal I}P_i(\ato T)$.
Let $B=\{\forall i\in\mathcal I:P_i(\ato T)\geq\tau'\}$ denote the event that all opinions in $\mathcal I$ are bad.
For every deterministic $m$, optional stopping at the bounded stopping time $T\wedge m$ gives $\E{\Phi(T\wedge m)}=\Phi(0)$.
On the event $B\cap\{T\leq m\}$, we have $\Phi(T\wedge m)=\prod_{i\in\mathcal I}P_i(\ato T)\geq(\tau')^k$.
Since $\Phi$ is nonnegative, it follows that
\begin{align}
    (\tau')^k\pr{B\cap\{T\leq m\}}
    \leq\E{\Phi(T\wedge m)}
    =\Phi(0).
    \label{eq:psiupper}
\end{align}
Finally, $T<\infty$ almost surely, so $\pr{B\cap\{T\leq m\}}\to\pr{B}$ as $m\to\infty$.
Taking this limit in \eqref{eq:psiupper} and applying \eqref{eq:psilower} yields
\begin{align}
    \pr{\forall i\in\mathcal I:P_i(\ato T)\geq\tau'}
    \leq\frac{\Phi(0)}{(\tau')^k}
    \leq\left(\frac{\Delta\hat s^2}{\tau'}\right)^k.
    \label{eq:martingalebound}
\end{align}
This proves the lemma.

\end{proof}

We now proceed to the proof of \Cref{lemma:Ctau}. Equipped with \Cref{lemma:pkbound}, we can bound the size of a witness via a counting argument and the union bound.

\medskip
\noindent\emph{Proof of \Cref{lemma:Ctau}:}
Fix an opinion $i$ and an integer $k\geq1$. We show that a witness containing at least $2\Delta^3 k$ opinions is unlikely.
If $i$ has been productive fewer than $\tau'$ times, its witness is the singleton $\{i\}$, so this event cannot occur.
Otherwise, every opinion in $W_i^{\tau'}(T)$ is bad, meaning that it has been productive at least $\tau'$ times.

For each connected set $C$ containing $i$, we construct a $4$-separated set $\mathcal I_C\subseteq C$ using a deterministic greedy algorithm.
All distances in this construction are measured in the ambient graph $G$.
For $R\subseteq O$, let $B_d(R)=\{v\in O:\dist{v}{R}\leq d\}$.
Fix an ordering of the opinions to break ties.
\begin{mdframed}
\textsc{\underline{GreedyNet($G,C$):}}
\begin{enumerate}[itemsep=4pt,topsep=2pt,parsep=0pt,partopsep=0pt,leftmargin=12pt]
\item Set $\mathcal I_C\gets\{i\}$.
\item While $C\setminus B_3(\mathcal I_C)\neq\emptyset$, do the following:
\begin{enumerate}
\item Choose the first node $j^*\in C$ in distance exactly four from $\mathcal I_C$.
\item Set $\mathcal I_C\gets\mathcal I_C\cup\{j^*\}$ and recompute $C\setminus B_3(\mathcal I_C)$.
\end{enumerate}
\end{enumerate}
\end{mdframed}
A suitable node exists whenever the loop continues: follow a path inside $C$ from $i$ to an uncovered node and take the first vertex whose distance from the current set is four.
Thus, the construction gives a $4$-separated set covering every vertex of $C$ within distance at most three.
A radius-three ball contains at most
\[
    1+\Delta+\Delta(\Delta-1)+\Delta(\Delta-1)^2\leq2\Delta^3
\]
vertices for $\Delta\geq1$. Therefore, $|\mathcal I_C|\geq |C|/(2\Delta^3)$.

If $|C|\geq2\Delta^3 k$, retain only the first $k$ nodes selected by the algorithm, including $i$.
This prefix is still $4$-separated. Moreover, every node except $i$ was added at distance exactly four from an earlier node.
Let $H_4$ be the graph joining opinions at distance exactly four in $G$. The prefix is connected in $H_4$, whose maximum degree is at most $\Delta^4$.

Let $\mathcal J_k(i)$ denote the family of $k$-vertex subsets containing $i$ that are $4$-separated in $G$ and connected in $H_4$.
A witness of size at least $2\Delta^3 k$ contains an all-bad member of this family. Consequently,
\begin{align}
    \pr{|W_i^{\tau'}(T)|\geq2\Delta^3 k}
    &\leq\pr{\bigcup_{I\in\mathcal J_k(i)}\{\forall j\in I:P_j(\ato T)\geq\tau'\}} \leq |\mathcal J_k(i)|
       \left(\frac{\Delta\hat s^2}{\tau'}\right)^k.
    \label{eqn:pkbound}
\end{align}
Here, the first inequality is an inclusion of events; it does not require the map from connected sets to prefixes to be injective.
It remains to count the possible prefixes. Every such prefix contains a spanning tree rooted at $i$, so we use the following bound.

\begin{restatable}{lem}{numbernets}
\label{lemma:numbernets}
Let $G = (V,E)$ be a graph of maximum degree at most $\Delta\geq1$, then each node $i \in V$ is part of at most $\left(2e^3\Delta\right)^k$ connected subtrees of size $k$.
\end{restatable}

\noindent We present the proof in Section \ref{sec:numbernets}.
Applying \Cref{lemma:numbernets} to $H_4$ gives
$|\mathcal J_k(i)|\leq(2e^3\Delta^4)^k$.
Since $\tau'=100\Delta^5\hat s^2$, \Cref{eqn:pkbound} now yields
\[
    \pr{|W_i^{\tau'}(T)|\geq2\Delta^3 k}
    \leq(2e^3\Delta^4)^k
       \left(\frac{1}{100\Delta^4}\right)^k
    \leq2^{-k}.
\]
As $W_i^{\tau'}(t)\subseteq W_i^{\tau'}(T)$ for every $t\geq0$, the same bound holds simultaneously over all times. This proves \Cref{lemma:Ctau} with $C=2$ and concludes the analysis. \qed

\clearpage

\section{Tools from Probability Theory \& Textbook Results}
\label{apx:tools}
This sections collects some (corollaries of) prior results used in the analysis.

\subsection{Bounds for Martingales}
We use the following standard versions of Doob's Optional Stopping Theorem.
This result is a classical cornerstone of martingale theory and can be found,
for instance, in standard textbooks on probability theory and stochastic
processes such as \cite{doob1953stochastic}.
It will be used repeatedly to justify that certain carefully constructed
processes preserve their expectation at suitable stopping times.

\begin{theorem}[Doob's Optional Stopping Theorem]
\label{thm:doob}

Let $(X(t))_{t\geq0}$ be a martingale with respect to a filtration $(\mathcal F(t))_{t\geq0}$, and let $T$ be a stopping time.
Suppose that at least one of the following conditions holds:
\begin{enumerate}
    \item $T$ is bounded by a deterministic constant.
    \item $T<\infty$ almost surely, and there is a deterministic constant $K<\infty$ such that $|X(t\wedge T)|\leq K$ almost surely for every $t\geq0$.
    \item $\E{T}<\infty$, and there is a deterministic constant $K<\infty$ such that $|X(t+1)-X(t)|\leq K$ almost surely for every $t\geq0$.
\end{enumerate}
Then
\[
    \E{X(T)}=\E{X(0)}.
\]
Under the same conditions, if $X$ is a submartingale, then $\E{X(T)}\geq\E{X(0)}$; if $X$ is a supermartingale, then $\E{X(T)}\leq\E{X(0)}$.

\end{theorem}

While Doob's theorem controls expectations at stopping times, we will also need
tail bounds that quantify how much a martingale (or supermartingale) can deviate
from its initial value over a fixed time horizon.
A classical tool for this purpose is Azuma's inequality, sometimes also referred
to as the Azuma--Hoeffding inequality, which provides sub-Gaussian concentration
for martingales with bounded differences; see, e.g.,
\cite{azuma1967weighted, hoeffding1963probability, mcdiarmid1998concentration}.

\begin{theorem}[Azuma's Ineqality]
\label{thm:azuma}

Assume that $(Y(t))_{t\geq0}$ is a supermartingale with
$|Y(t+1)-Y(t)|\leq b$ almost surely.
Then, for every $\lambda>0$ and integer $T\geq1$,
\[
    \pr{\max_{0\leq t\leq T}(Y(t)-Y(0))\geq\lambda}
    \leq\exp\left(-\frac{\lambda^2}{2Tb^2}\right).
\]
In particular, the same bound holds for $\pr{Y(T)-Y(0)\geq\lambda}$.

\end{theorem}

In some situations, however, it is not the value of the process at a fixed time
that is of interest, but rather the maximum deviation attained over an entire
time interval.
Kolmogorov's maximal inequality addresses exactly this question for martingales
by relating the maximum absolute deviation to the second moment at the terminal
time.
This inequality is classical folklore and can be found, for example, in
\cite{kolmogorov1933grundbegriffe}.

\begin{theorem}[Kolmogorov's maximal inequality]
\label{thm:kolmogorov}
Let $(M(t))_{t\ge 0}$ be a real-valued martingale with $M(0)=0$ and
$\E{M(t)^2}<\infty$.
Then for all $t\ge 1$ and all $\lambda>0$,
\[
\pr{\max_{0\le s\le t} |M(s)| \ge \lambda}
\le
\frac{\E{M(t)^2}}{\lambda^2}.
\]
\end{theorem}

The following lemma bounds a martingale in terms of its number of nonzero increments.
The exponential-supermartingale construction is a special case of Lemma~1.4 in~\cite{deLaPenaKL04}.
We include a proof for completeness.

\begin{lemma}[Tail bound in terms of nonzero increments]
\label{lemma:nonzero-increments}
Let $(X_t)_{t\geq0}$ be a martingale with respect to a filtration $(\mathcal F_t)_{t\geq0}$, with $X_0=0$ and $X_t-X_{t-1}\in\{-1,0,1\}$ almost surely for all $t\geq1$.
Define
\[
    V_t=\sum_{s=1}^t(X_s-X_{s-1})^2,
\]
which counts the nonzero increments before or at time $t$.
Then, for every deterministic $a,v>0$,
\[
    \pr{\exists t\geq0:X_t\geq a,\ V_t\leq v}
    \leq\exp\left(-\frac{a^2}{2v}\right).
\]
\end{lemma}

\begin{proof}
Write $D_t=X_t-X_{t-1}$ and $p_t=\pr{|D_t|=1\mid\mathcal F_{t-1}}$.
Since $X_t$ is a martingale, $\E{D_t\mid\mathcal F_{t-1}}=0$.
As $D_t\in\{-1,0,1\}$, this implies
\[
    \pr{D_t=1\mid\mathcal F_{t-1}}
    =\pr{D_t=-1\mid\mathcal F_{t-1}}=\frac{p_t}{2}.
\]
For a fixed $\lambda>0$, define
\[
    M_t=\exp\left(\lambda X_t-\frac{\lambda^2}{2}V_t\right).
\]
Conditioning on $\mathcal F_{t-1}$ gives
\begin{align*}
    \E{M_t\mid\mathcal F_{t-1}}
    &=M_{t-1}\left(1-p_t+p_t e^{-\lambda^2/2}
                     \frac{e^\lambda+e^{-\lambda}}{2}\right) \leq M_{t-1}.
\end{align*}
Here, we used
\[
    \frac{e^\lambda+e^{-\lambda}}{2}
    =\sum_{j=0}^\infty\frac{\lambda^{2j}}{(2j)!}
    \leq\sum_{j=0}^\infty\frac{\lambda^{2j}}{2^j j!}
    =e^{\lambda^2/2}.
\]
Thus, $(M_t)_{t\geq0}$ is a nonnegative supermartingale with $M_0=1$; integrability follows from $M_t\leq e^{\lambda t}$ at every fixed time $t$.

Let $\tau_{a,v}=\inf\{t\geq0:X_t\geq a,\ V_t\leq v\}$, with $\tau_{a,v}=\infty$ if this event never occurs.
For every deterministic integer $h\geq1$, optional stopping at the bounded stopping time $\tau_{a,v}\wedge h$ gives
$\E{M_{\tau_{a,v}\wedge h}}\leq1$.
On $\{\tau_{a,v}\leq h\}$, we have
$M_{\tau_{a,v}\wedge h}=M_{\tau_{a,v}}\geq e^{\lambda a-\lambda^2v/2}$.
Since $M$ is nonnegative, it follows that
\[
    \pr{\tau_{a,v}\leq h}\leq e^{-\lambda a+\lambda^2v/2}.
\]
Letting $h\to\infty$ gives the same bound for $\pr{\tau_{a,v}<\infty}$ by continuity of probability.
Finally, choosing the deterministic value $\lambda=a/v$ yields the desired bound.
\end{proof}

\subsection{Tail-Bounds for Weakly-dependent Variables}
In several parts of our analysis, we need concentration bounds for sums of
random variables that are not fully independent.
We use two bounds that capture different forms of weak dependence.
The first is an extension of a Chernoff bound due to Scheideler
\cite{scheidlerhabil}, based on the work of Schmidt, Siegel, and
Srinivasan~\cite{DBLP:journals/siamdm/SchmidtSS95}.
It applies to bounded random variables whose joint moments satisfy a suitable
global upper bound.
The second bound concerns binary random variables for which the required
product inequality is available only for sets that are independent in a given
graph.
By coloring this graph and combining the color classes using Hölder's
inequality, we obtain a lower-tail bound whose exponent loses only a factor
equal to the chromatic number. A similar has been used before by Janson

For the first bound, our starting point is a theorem by Scheideler \cite{scheidlerhabil} that itself is heavily influenced by \cite{DBLP:journals/siamdm/SchmidtSS95}.
In particular, Scheideler observed that the Chernoff bounds presented
in~\cite{DBLP:journals/siamdm/SchmidtSS95} for $k$-wise independent variables hold for a much wider class of random
variables.
In particular, $k$-wise independence is not required; instead, the variables must
satisfy a specific technical inequality that is sufficient for the proof
technique of~\cite{DBLP:journals/siamdm/SchmidtSS95} to go through.

We use the corresponding upper-tail bound under a condition on elementary symmetric polynomials.
For $a_1,\ldots,a_k\in\mathbb R$ and $1\leq l\leq k$, define
\[
    P_l(a_1,\ldots,a_k)
    \coloneqq\sum_{1\leq i_1<\cdots<i_l\leq k}\prod_{j=1}^l a_{i_j}.
\]
The indices are distinct. In particular, $P_l(p,\ldots,p)=\binom kl p^l$.
\begin{theorem}[Chernoff bound under a moment condition]
\label{thm:sss+}
Let $X_1,\ldots,X_k\in[0,1]$ and $0<p<1$. Suppose that, for every $1\leq l\leq k$,
\[
    \E{P_l(X_1,\ldots,X_k)}\leq\binom kl p^l.
\]
Set $S=\sum_{i=1}^kX_i$ and $\mu=pk$. Then, for every $\delta\geq1$,
\[
    \pr{S\geq\delta\mu}
    \leq\exp\bigl(-\mu(\delta\log\delta-\delta+1)\bigr).
\]
\end{theorem}
\begin{proof}
For $x\in[0,1]$ and $\theta\geq0$, convexity gives
$e^{\theta x}\leq1+(e^\theta-1)x$.
Expanding the product and applying the assumed bounds on the elementary symmetric polynomials yields
\begin{align*}
    \E{e^{\theta S}}
    &\leq\sum_{l=0}^k(e^\theta-1)^l\E{P_l(X_1,\ldots,X_k)} \leq\bigl(1+p(e^\theta-1)\bigr)^k
    \leq\exp\bigl(\mu(e^\theta-1)\bigr),
\end{align*}
where $P_0=1$. Markov's inequality with $\theta=\log\delta$ gives the claim; for $\delta=1$ the bound is immediate.
\end{proof}

We extend this bound to variables that are bounded by $b$. Namely,
\begin{theorem}
\label{thm:scheidelbound}

Let $S=\sum_{i=1}^kX_i$, where $X_i\in[0,b]$ and $b>0$.
Suppose that $0<p<1$ and $x_1,\ldots,x_k\in[0,b]$ satisfy, for every subset $I\subseteq[k]$,
\begin{align}
    \E{\prod_{i\in I}X_i}\leq p^{|I|}\prod_{i\in I}x_i.
    \label{eq:scheidelcond}
\end{align}
Set $\mu=p\sum_{i=1}^kx_i$. If $\mu>0$, then for every $\delta\geq1$,
\[
    \pr{S\geq\delta\mu}
    \leq\exp\left(-\frac{\mu}{b}(\delta\log\delta-\delta+1)\right).
\]
In particular, for $\delta\geq e^2$, this probability is at most
$\exp(-\delta\mu/b)$. If $\mu=0$, then $S=0$ almost surely.

\end{theorem}
\begin{proof}

We use the same argument as for \Cref{thm:sss+}, keeping the individual values $x_i$ in the product.
Let $X'_i=X_i/b\in[0,1]$. For $\theta\geq0$, convexity and \Cref{eq:scheidelcond} give
\begin{align*}
    \E{e^{\theta S/b}}
    &\leq\E{\prod_{i=1}^k\left(1+(e^\theta-1)X_i/b\right)} =\sum_{I\subseteq[k]}\left(\frac{e^\theta-1}{b}\right)^{|I|}
       \E{\prod_{i\in I}X_i}\\
    &\leq\prod_{i=1}^k\left(1+\frac{p x_i}{b}(e^\theta-1)\right)
    \leq\exp\left(\frac{\mu}{b}(e^\theta-1)\right).
\end{align*}
Here, all coefficients are nonnegative, so the product assumption applies term by term. In particular, the parameter $\mu$ uses the sum over all $k$ values $x_i$.
For $\mu>0$ and $\delta>1$, Markov's inequality with $\theta=\log\delta$ gives
\[
    \pr{S\geq\delta\mu}
    \leq\exp\left(-\frac{\delta\mu}{b}\log\delta
                        +\frac{\mu}{b}(\delta-1)\right).
\]
This is the desired bound. If $\delta\geq e^2$, then
$\delta\log\delta-\delta+1\geq\delta$.
Finally, if $\mu=0$, the singleton cases of \Cref{eq:scheidelcond} imply $\E{X_i}=0$ for every $i$, so all variables vanish almost surely.

\end{proof}

The preceding bound is useful when the required product estimate holds for
arbitrary subsets of the variables.
For the number of surviving opinions, however, we obtain the relevant product
estimate only for opinions that are sufficiently far apart in the opinion
graph.
The following lemma turns such local product estimates into a global
lower-tail bound.
More precisely, we introduce a graph $H$ whose independent sets are exactly
the sets for which the product estimate is available.
A proper coloring of $H$ partitions the variables into classes on which the
standard moment-generating-function argument applies.
The dependence between different color classes costs a factor of
$\chi(H)$ in the exponent.

\graphchernoff*

\begin{proof}
We follow the textbook proof for the Chernoff bound, but with the restriction of the graph $H$.
Fix a proper coloring of $H$ with color classes
$C_1,\ldots,C_\kappa$, and let $Y_i:=1-X_i$ be the indicator that $X_i=0$.
Every subset of a color class is an independent set of $H$.
Note that for a binary random variable $Y \in \{0,1\}$ and every
$\theta\geq 0$, we have
\begin{align*}
    \exp{(\theta Y)} = \begin{cases}
        1 & \textbf{if } Y =0\\
        \exp{\theta} & \textbf{if } Y =1.
    \end{cases}
\end{align*}
Therefore $ \exp{(\theta Y)} = 1+(e^\theta-1)Y$.
Thus, for every
$\theta\geq 0$, combining this fact and
\eqref{eq:local-product-assumption} gives
\begin{align*}
    \E{
        \exp\left(\theta\sum_{i\in C_r}Y_i\right)
    }
    &=
    \E{
        \prod_{i\in C_r} \exp\left(\theta Y_i\right)
    }
    =
    \E{
        \prod_{i\in C_r}\bigl(1+(e^\theta-1)Y_i\bigr)
    }
    =
    \sum_{A\subseteq C_r}
        (e^\theta-1)^{|A|}
        \E{\prod_{i\in A}Y_i}\\
    &\leq
    \sum_{A\subseteq C_r}
        (e^\theta-1)^{|A|}
        \prod_{i\in A}(1-p_i)
    =
    \prod_{i\in C_r}
        \bigl(1+(e^\theta-1)(1-p_i)\bigr).
\end{align*}
Here, we also used that $\prod_{i \in C_r} (1 + Y'_i) = \sum_{A \subseteq C_r}\prod_{i \in A} Y'_i$ where $Y'_i = (e^\theta-1)Y_i$.
All coefficients in the expansion are non-negative, which is the reason why
the local product assumption can be applied term by term.
Since $X_i=1-Y_i$, the preceding inequality implies
\begin{align*}
    \E{
        \exp\left(-\theta\sum_{i\in C_r}X_i\right)
    }
    &\leq
    \prod_{i\in C_r}
        \bigl(1-p_i(1-e^{-\theta})\bigr)
    \leq
    \exp\left(
        -(1-e^{-\theta})\sum_{i\in C_r}p_i
    \right),
\end{align*}
where we used $1-z\leq e^{-z}$ in the last step.
Next, we consider all $\kappa$ color classes.
The color classes are not necessarily independent of each other.  We handle
this using Hölder's inequality.
For every $\lambda\geq 0$,
\begin{align*}
    \E{e^{-\lambda Z}}
    &=
    \E{
        \prod_{r=1}^{\kappa}
        \exp\left(-\lambda\sum_{i\in C_r}X_i\right)
    }
    \leq
    \prod_{r=1}^{\kappa}
    \left(
        \E{
            \exp\left(-\kappa\lambda
                \sum_{i\in C_r}X_i
            \right)
        }
    \right)^{1/\kappa}
    \leq
    \exp\left(
        -\frac{1-e^{-\kappa\lambda}}{\kappa}\mu
    \right).
\end{align*}

If $ Z\leq (1-\delta)\mu$, then
$e^{-\lambda Z}\geq e^{-\lambda(1-\delta)\mu}$ for any $\lambda \geq 0$.
Markov's inequality therefore gives
\begin{align*}
    \pr{Z\leq(1-\delta)\mu} &\leq \pr{\lambda Z\leq \lambda(1-\delta)\mu} = \pr{e^{-\lambda Z}\geq e^{-\lambda(1-\delta)\mu}} \leq \frac{\E{e^{-\lambda Z}}}{e^{-\lambda(1-\delta)\mu}}
\end{align*}
Using our previously computed bound for $\E{e^{-\lambda Z}}$ then gives
\begin{align*}
    \pr{Z\leq(1-\delta)\mu} \leq \frac{e^{-\frac{1-e^{-\kappa\lambda}}{\kappa}\mu} }{e^{-\lambda(1-\delta)\mu}}
    =
    \exp\left(
        \lambda(1-\delta)\mu
        -\frac{1-e^{-\kappa\lambda}}{\kappa}\mu
    \right).
\end{align*}
It remains to find the right choice for $\lambda$.
For $\delta=0$, the bound is trivial. For $\delta=1$, letting $\lambda\to\infty$ in the preceding bound gives $\pr{Z=0}\leq\exp(-\mu/\kappa)\leq\exp(-\mu/(2\kappa))$.
For $0<\delta<1$, we choose $\lambda=-\frac{1}{\kappa}\log(1-\delta)$ so that  $e^{-\kappa\lambda}=1-\delta$.
Substituting this value proves
\begin{align*}
    \pr{Z\leq(1-\delta)\mu}
    \leq
    \exp\left(
        -\frac{\mu}{\kappa}
        \bigl(\delta+(1-\delta)\log(1-\delta)\bigr)
    \right).
\end{align*}
As $-\log(1-\delta) \geq \delta$ for all $\delta \in (0,1)$, it holds
\begin{align*}
      \delta+(1-\delta)\log(1-\delta)
    = \int_0^\delta -\log(1-x)\,dx
    \geq \int_0^\delta x\,dx
    = \frac{\delta^2}{2},
\end{align*}
which proves the theorem.
\end{proof}

\subsection{Lower Bound for the convergene of Gambler's Ruin.}

\begin{lemma}
\label{lem:gr_tail_lower}
Consider the fair gambler's ruin chain $(X(t))_{t\ge 0}$ on
$\{0,1,\dots,n\}$ with absorbing states $\{0,n\}$.
Let
\[
T \coloneqq \min\{t\ge 0 : X(t) \in \{0,n\}\}
\]
be the stopping time.
If $X(0)=1$, then there exists an absolute constant
$c>0$ such that $\pr{T \ge c n^2} = \Omega\!\left(\frac{1}{n}\right)$.
\end{lemma}

\begin{proof}
\newcommand{\prm}[1]{\mathbf{Pr}_m\left[#1\right]}
For $a\in\{0,1,\dots,n\}$ let
$\tau_a \coloneqq \min\{t\ge 0 : X(t)=a\}$.
Set $m\coloneqq \lfloor n/2 \rfloor$.

Since $\tau_0\wedge\tau_m$ is almost surely finite and the stopped martingale is bounded between $0$ and $m$, applying optional stopping at $(\tau_0\wedge\tau_m)\wedge r$ and then bounded convergence as $r\to\infty$ yields

\[
\E{X(\tau_0\wedge \tau_m)} = \E{X(0)} = 1.
\]
Moreover, $X(\tau_0\wedge \tau_m)\in\{0,m\}$, hence
\[
1 = \E{X(\tau_0\wedge \tau_m)}
  = m \cdot \pr{\tau_m < \tau_0},
\]
and therefore
\[
\pr{\tau_m < \tau_0} = \frac{1}{m}.
\]
Let now $\prm{\cdot}$ be shorthand for $\pr{\cdot \growingmid X(0)=m}$.
Using the strong Markov property at time $\tau_m$ we obtain, for any $t\ge 0$,
\[
\pr{T \ge t}
\;\ge\;
\pr{\tau_m < \tau_0}\cdot \prm{T \ge t}
\;=\;
\frac{1}{m}\cdot \prm{T \ge t}.
\]
It remains to show that $\prm{T \ge t}$ is bounded below by a positive
constant for $t=\Theta(n^2)$.

Consider a simple random walk $(S(t))_{t\ge 0}$ on $\mathbb{Z}$ with $S(0)=0$.
We couple it to the gambler's ruin chain started at $m$ so that
$S(t) = X(t) - m$ holds for all $t<T$.
If $T \le t$, then by time $t$ the chain must have hit either $0$ or $n$,
and thus $|X(s)-m|\ge m$ for some $s\le t$.
Consequently,
\[
\prm{T \le t}
\;\le\;
\pr{\max_{0\le s\le t} |S(s)| \ge m}.
\]
By Kolmogorov's maximal inequality (see \Cref{thm:kolmogorov}) applied to the martingale $(S(t))_{t\ge 0}$,
\[
\pr{\max_{0\le s\le t} |S(s)| \ge m}
\;\le\;
\frac{\E{S(t)^2}}{m^2}.
\]
Since $\E{S(t)^2}=t$ for simple random walk on $\mathbb{Z}$, we obtain
\[
\prm{T \le t} \le \frac{t}{m^2}.
\]
Choosing $t \coloneqq \frac{1}{4}m^2$ gives $\prm{T \le t}\le 1/4$, and hence
\[
\prm{T \ge t}\ge \frac{3}{4}.
\]
Putting everything together,
\[
\pr{\!\left(T \ge \tfrac{1}{4}m^2\right)}
\;\ge\;
\frac{1}{m}\cdot \frac{3}{4}
\;=\;
\Omega\!\left(\frac{1}{n}\right).
\]
Since $m=\Theta(n)$, this implies $\pr{T \ge c n^2}=\Omega(1/n)$ for some
absolute constant $c>0$.
\end{proof}

\subsection{Number of Subtrees of Size $k$}
\label{sec:numbernets}
\numbernets*
\begin{proof}
We use the classical tree counting approach via the Catalan numbers \cite{catalan1838note}.
There are $C_{k-1}$ ordered rooted tree shapes on $k$ vertices, where
\[
    C_{k-1}=\frac{1}{k}\binom{2k-2}{k-1}\leq4^{k-1}.
\]
Fix one such shape and map its root to $i$.
We embed the remaining vertices in breadth-first order. Each non-root vertex must be mapped to a neighbor of its parent, so it has at most $\Delta$ choices.
This gives at most $\Delta^{k-1}$ maps. Some maps identify vertices and do not yield a subtree, but counting these maps only increases our upper bound.
Conversely, every rooted subtree can be represented by at least one ordered rooted tree shape and an injective map of this form.
Therefore, for $\Delta\geq1$, the number of subtrees is at most
\[
    C_{k-1}\Delta^{k-1}\leq(4\Delta)^{k-1}
    \leq(2e^3\Delta)^k.
\]
This also covers $k=1$. The same bound applies to connected vertex sets, since each such set has a spanning tree rooted at $i$.
\end{proof}

\clearpage

\section{A New Analysis for Asynchronous Pull Voting on Graphs}
\label{apx:voting}
In this section, we show that our potential can also be used to analyze a simple variant of the standard pull voting process on regular graphs.
Consider the following pull voting process \texttt{AsyncVoting} on a connected $d$-regular graph $G=(V,E)$ with $n=|V|\geq2$ and $E\subseteq\binom V2$:
In each step $t \geq 0$, a pair of agents $\{u,v\} \in E$ is picked uniformly at random and then, with probability $\frac{1}{2}$ each, $u$ adopts the opinion of $v$ or vice versa.
On a regular graph, this is equivalent to first picking $u$ uniformly at random and then letting this agent adopt the opinion of a neighbor $v$ chosen uniformly at random.
We show that this process converges in $O(\tfrac{n^2}{\Phi_G}\log n)$ steps, w.h.p., where $\Phi_G$ is the conductance of $G$ defined as:
\begin{align*}
    \Phi_G \coloneqq \min_{\substack{S \subset V \\ 0<|S| \leq n/2}} \frac{|\{\{u,v\} \in E \mid u \in S, v \in V \setminus S \}|}{d \cdot |S|}.
\end{align*}
Formally, we show
\begin{theorem}
Let $G=(V,E)$ be a connected $d$-regular graph on $n\geq2$ vertices with conductance $\Phi_G$.
Then, for any initial configuration \texttt{AsyncVoting} converges in $O(\nicefrac{n^2}{\Phi_G}\log n)$ steps w.h.p.
\end{theorem}
Cooper and Rivera~\cite[Example 11]{CooperR16} prove an expected bound of $O(n^2/\Phi_G)$ for asynchronous pull voting on regular graphs. Applying Markov's inequality in consecutive blocks gives the same high probability bound as above. Berenbrink et al.~\cite{BerenbrinkGKM16} prove a related bound for synchronous lazy voting, where time is measured in rounds in which all agents update. Both papers use variance-based arguments.
Our potential offers an alternative and arguably even simpler proof.

\begin{proof}

First, assume that there are two opinions $0$ and $1$.
We denote their supports by $S_0(t)$ and $S_1(t)$ and define
$\Psi(\Vec S(t))=S_0(t)S_1(t)$.
Let $\mathcal F(t)$ be the natural filtration, which records the opinions of all vertices and the previous interactions.
Write $\mathcal P(t)$ for the indicator that step $t$ selects two agents with different opinions.
The same calculation as in \Cref{lemma:new_potential} gives
\[
    \E{\Psi(\Vec S(t+1))\mid\mathcal F(t)}
    =\Psi(\Vec S(t))-\pr{\mathcal P(t)=1\mid\mathcal F(t)}.
\]
Before consensus, let $C_t\subset V$ be the smaller opinion class, and write $s=|C_t|\in\{1,\ldots,\lfloor n/2\rfloor\}$.
Let $\partial C_t$ denote the edges with one endpoint in $C_t$ and one outside.
Since $|E|=dn/2$, the definition of conductance gives
\begin{align*}
    \pr{\mathcal P(t)=1\mid\mathcal F(t)}
    &=\frac{|\partial C_t|}{|E|}
      =\frac{2|\partial C_t|}{dn} \geq\frac{2\Phi_Gs}{n}
      \geq\frac{2\Phi_G}{n^2}\Psi(\Vec S(t)),
\end{align*}
where the last inequality uses $\Psi(\Vec S(t))=s(n-s)\leq ns$.
Consequently,
\[
    \E{\Psi(\Vec S(t+1))\mid\mathcal F(t)}
    \leq\left(1-\frac{2\Phi_G}{n^2}\right)\Psi(\Vec S(t)).
\]
This also holds after consensus, when the potential remains zero.
Taking expectations and iterating yields
\[
    \E{\Psi(\Vec S(t))}
    \leq\Psi(\vec s_0)\left(1-\frac{2\Phi_G}{n^2}\right)^t
    \leq n^2\exp\left(-\frac{2\Phi_Gt}{n^2}\right).
\]
For any fixed $c>1$, choose
$t^*=\lceil(c+3)n^2\log n/(2\Phi_G)\rceil$.
Then the expected potential is at most $n^{-(c+1)}$.
As the potential is a positive integer before consensus, Markov's inequality bounds the probability that the binary process has not reached consensus by time $t^*$ by $n^{-(c+1)}$.

Finally, consider an arbitrary initial configuration with $q\leq n$ opinions.
For each initial opinion $i$, assign binary opinion $1$ to its agents and $0$ to all other agents.
Under the same edge selections and copying directions, this projection follows exactly the binary process just analyzed.
A union bound over the $q$ projections shows that, with probability at least $1-n^{-c}$, all projections have reached consensus by time $t^*$.
At that time, choose any opinion still present. Its projection contains a $1$ and has reached consensus, so every agent holds that opinion.
Thus, the original process has also reached consensus, proving the theorem.

\end{proof}

\printbibliography

\end{document}